\PassOptionsToPackage{hidelinks}{hyperref}
\documentclass[runningheads]{llncs}
\usepackage[T1]{fontenc}
\usepackage{amsmath,amssymb}
\usepackage{enumitem}
\setlist{nosep}
\usepackage{url}
\usepackage[hidelinks]{hyperref}
\usepackage[linesnumbered,ruled,vlined]{algorithm2e}
\usepackage{braket}
\usepackage{color}
\usepackage{dsfont}
\usepackage{caption}
\usepackage{float} 
\newenvironment{proofsketch}
{\noindent\textbf{Proof sketch.}}
{\hfill$\square$}

\usepackage{mathtools}
\usepackage{amsfonts}
\usepackage{graphicx}
\usepackage{textcomp}
\title{Measure Many Quantum Finite Automata on Infinite Words}
\author{
Abhisek Midya\inst{1} \and
A Baskar\inst{2}
}

\institute{
Department of Computer Science and Engineering, Dayananda Sagar University, Bengaluru, India\\
\email{abhisekmidyacse@gmail.com} \and
Department of CS \& IS, BITS Pilani, K.K. Birla Goa Campus, Goa\\
\email{abaskar@goa.bits-pilani.ac.in}}
\begin{document}

\maketitle
\begin{abstract}
We define a quantum computational model, called
\emph{Measure-Many Quantum Büchi Automata} (MMQBA), which extends
\emph{Measure-Many Quantum Finite Automata} (MMQFA) to the infinite-word
setting with a Büchi acceptance condition. In MMQBA, the quantum state
evolves through unitary transformations followed by repeated projective
measurements. An infinite word is accepted with respect to a cutpoint
$p \in (0,1]$ if (i) the run visits accepting states infinitely often,
(ii) the limiting cumulative acceptance probability is at least $p$, and
(iii) the limiting cumulative rejection probability is strictly less
than $p$. We formalize the semantics of MMQBA and establish a
language-theoretic characterization showing that MMQBA languages are
precisely of the form $\lim(L(M,p))$ for an MMQFA $M$. We further develop
a structural analysis of the non-halting subspace and prove that certain
$\omega$-languages cannot be recognized by MMQBA. Finally, we show that
the class of MMQBA-recognizable languages is not closed under union, intersection and complement.
\end{abstract}
\keywords{
Quantum finite automata, Infinite words, Büchi acceptance condition}
\section{Introduction}
Quantum finite automata (QFA) combine principles of quantum mechanics with classical finite automata and provide a theoretical model for finite-memory quantum computation. In one-way QFA models (1QFA), the input head moves only to the right and measurements may be performed either only at the end of the computation (measure-once) or after reading each input symbol (measure-many). Although several variants of QFA achieve exponential state savings compared with classical automata, many bounded-error models recognize only a proper subset of regular languages.

Büchi automata \cite{richard1962decision} are classical acceptors of
infinite words ($\omega$-words) and play a central role in the analysis
of systems with non-terminating behavior. They are widely used in
formal verification and model checking. Several related acceptance
conditions have been studied, including Muller, Rabin, and Streett
automata \cite{thomas1990automata}, as well as probabilistic variants
of $\omega$-automata \cite{baier2005recognizing,katoen2016probabilistic}. For a broad overview of quantum finite automata and related models in quantum computing, we refer the reader to the survey by Ambainis and Yakaryılmaz \cite{ambainis2021automata}.

Quantum automata over finite words have been investigated for more than two
decades; see, for example,
\cite{ambainis19981,brodsky2002characterizations,kondacs1997power,moore2000quantum}.
However, quantum automata on infinite words have received relatively limited
attention. Preliminary definitions of quantum Büchi, Streett, and Rabin
automata were introduced in~\cite{dzelme2010quantum}. In particular,
\cite{dzelme2010quantum} studied measure-once quantum Büchi automata (MOQBA)
and proposed the notion of measure-many quantum Büchi automata (MMQBA).
Their acceptance criterion requires the cumulative acceptance probability to
converge to~1. Although this implicitly ensures infinitely many accepting
visits, it does not restrict the rejection behavior of the computation.

In this paper we introduce a refined definition of measure-many quantum
Büchi automata (MMQBA). In our model, an infinite word is accepted if
(i) accepting states are visited infinitely often, (ii) the cumulative
acceptance probability converges to at least a cutpoint $p$, and
(iii) the cumulative rejection probability remains strictly below $p$.
This formulation provides a more robust semantics by preventing cases
where Büchi recurrence holds while significant rejection probability
accumulates. We first establish a limit characterization connecting MMQBA-recognizable $\omega$-languages with languages recognized by measure-many quantum finite automata (MMQFA). This characterization provides a structural connection between finite-word and infinite-word quantum automata, analogous to classical limit constructions for $\omega$-languages, while preserving the distinct infinite-computation behavior captured by MMQBA. We then prove a
non-recognizability result showing that certain languages cannot be
recognized by MMQBA with cutpoint greater than $\tfrac12$, and further
investigate structural properties of the non-halting subspace as well as
closure properties of the model. Our approach differs from the quantum Büchi automaton model of
\cite{wang2024quantum}, which is based on disturbing or non-disturbing
measurements in a measure-once setting rather than cumulative cutpoint
semantics. In contrast, our model incorporates repeated measurements and
cumulative acceptance/rejection probabilities, making it closer in spirit
to measure-many quantum finite automata and better suited for studying
limit behavior of infinite computations. Similar to classical Büchi
automata, acceptance depends on infinite recurrence behavior along
infinite runs; however, in the quantum setting this behavior is further
governed by cumulative acceptance and rejection probabilities arising
from repeated measurements.

The results developed herein contribute to the theoretical foundations of
quantum automata operating on infinite words. Such models provide a
mathematical framework for studying quantum computational processes that
evolve indefinitely under repeated measurements. Understanding these
foundations may be relevant to areas such as the formal analysis and
verification of quantum systems, the semantics and verification of quantum programs as well as quantum cryptographic protocols \cite{brunet2004dynamic,chadha2006reasoning,feng2007proof}, and the formal analysis of quantum games \cite{eisert1999quantum,gutoski2007toward}.

\par In Section~2, we recall the framework of measure-many quantum finite automata on finite words, which forms the basis for our later development. 
Section~3 presents the formal definition of MMQBA together with an example. Section 4, 5 discuss the structural properties and language characterization of MMQBA respectively. Section~6 investigates closure properties. Finally, Section~7 concludes with a summary and future research directions. 

Additional examples and detailed proofs are provided in the Appendix.

\section{Preliminaries}
\begin{definition}$MMQFA$\cite{ambainis19981}\\
$MMQFA$ is a tuple $M=(Q, \Sigma,\delta,q_0,Q_{acc},Q_{rej})$ where 
\begin{itemize}
\item
$Q$ is a finite set of states, 
\item 
$\Sigma$ is an input alphabet,
\item $\delta:Q \times \Gamma \times Q \rightarrow \mathds{C}_{[0,1]}$ is the transition function 
\item $q_0 \in Q$ is a starting state and 
\item $Q_{acc} \subseteq Q$ is a set of accepting states and 
\item $Q_{rej} \subseteq Q$ is a set of rejecting states. 
\end{itemize}
\end{definition}
The states in $Q_{acc}$ and $Q_{rej}$ are referred to as halting states and the
states in $Q_{non} = Q-(Q_{acc} \cup Q_{rej})$ are referred to as non-halting states. Let $\#, \$ \notin \Sigma$ and $\Gamma = \Sigma \cup \{\#, \$\}$ where $\#, \$$ are left and right end markers respectively.
\par Recall notation $l_2(Q)$ (the space mappings from $Q$ to $\mathbb{C}$ with $l_2$ norm), and emphasize that superpositions of M have norm 1.
. $\ket{q}$ represents the unit vector for $q \in Q$, which has a value of 1 at $q$ and $0$ elsewhere. It is possible to express every element of $l_2(Q)$ as a linear combination of vectors $\ket{q}$. We will use $\psi$ to
denote elements of $l_2(Q)$ and $S_U$ to denote the span of basis states in $U \subseteq Q$. In other words $S_U=span(\{\ket{q}: q \in U\})$.
\par We define $\delta : Q \times \Gamma \times Q \longrightarrow \mathds{C}_{[0,1]}$. The amplitude of $\ket{q_2}$ in the superposition of states to which $M$ proceeds from $\ket{q_1}$ after reading $\sigma$ is represented by the value $\delta(q_1, \sigma, q_2)$. For $\sigma \in \Gamma$,
$V_{\sigma}$ is a linear transformation on $l_2(Q)$ defined by
$$V_{\sigma}(\ket{q_1})=\sum_{q_2\in Q}\delta(q_1, \sigma, q_2)\ket{q_2}$$
\par We require all \( V_{\sigma} \) to be unitary\footnote{A matrix \( U \) is unitary if \( U^\dagger U = I \), where \( U^\dagger \) is the conjugate transpose of \( U \). Unitary matrices preserve the norm of quantum states.}. The computation of  $M$ starts in the superposition $\ket{q_0}$. Then transformations corresponding
to the left end marker \#, the letters$(\sigma)$ of the input word $x$ and the right end marker $\$$ are applied. The transformation corresponding to $\sigma\in \Gamma$ consists of the following steps.
\begin{itemize}
\item 
First, $V_{\sigma}$ is applied. The new superposition $\psi'$ is $V_{\sigma}(\psi)$, where $\psi$ is the superposition before this step.

\item 
Then, $\psi'$ is measured with respect to the observable 
$S_{acc} \oplus S_{rej} \oplus S_{non}$, where $\oplus$ denotes the orthogonal direct sum of subspaces. 
Here,
$S_{acc} = \mathrm{span}(\{\ket{q} : q \in Q_{acc}\})$, 
$S_{rej} = \mathrm{span}(\{\ket{q} : q \in Q_{rej}\})$, and 
$S_{non} = \mathrm{span}(\{\ket{q} : q \in Q_{non}\})$. 

The input is accepted with probability $\|P_{acc}\psi'\|^2$, rejected with probability $\|P_{rej}\psi'\|^2$, and the computation continues from $P_{non}\psi'$ with probability $\|P_{non}\psi'\|^2$. $P_p, p \in \{acc, rej, non\}$ the projection operator into the subspace
$S_p$, and and let $P_{\text{halt}} := P_{acc} + P_{rej}$ denote the projection onto the halting subspace.

Note that $P_{non}\psi'$ is generally not normalized; its squared norm represents the remaining non-halting probability.
\end{itemize}
In order to define formally the overall probability with which an input is accepted (rejected)
by a $MMQFA$ M, we define the set $V_M = l_2(Q) \times C \times C$ of so-called “total states” of $M$, that
will be used only with the following interpretation. $M$ is at any time during the computation
in the state $(\psi, p_{a}, p_{r})$ if so far in its computation $M$ accepted the input with probability $p_{a}$,
rejected with probability $p_{r}$ and neither with probability $1-p_{a}-p_{r} = ||\psi||^2$, and $\ket{\psi}$ is its  
current, unnormalized state. For each $\sigma \in \Gamma$ the evolution of $M$, with respect to the total
state, on an input $\sigma$ is given by the operator $T_{\sigma}$ defined in the following way:
$$T_{\sigma}(\psi,p_{a},p_{r})\longrightarrow (P_{non}V_{\sigma}\psi,p_{a}+||P_{acc}V_{\sigma}\psi||^2,p_{r}+||P_{rej}V_{\sigma}\psi||^2)$$. \\
For $x=\sigma_1\sigma_2 \ldots \sigma_n \in \Sigma^*$ let $T_{\#x\$}=T_{\$}T_{\sigma_n}T_{\sigma_{n-1}}\ldots T_{\sigma_{1}}T_{\#}$. If $T_{\#x\$}(\ket{q_0}, 0, 0) = (\psi, p_{a}, p_{r})$, then we say that $M$ accepts $x$ with probability $p_a$ and rejects with probability $p_r$.  
We say that $L(M,p)$ is the language accepted by $M$ with cutpoint $p$, that is, $L(M,p) := \{\, x \in \Sigma^* \mid \text{$M$ accepts $x$ with probability at least $p$} \,\}.$
\section{Measure Many Quantum Büchi Automata (MMQBA)} 
\begin{definition}
$MMQBA$ is a tuple $M=(Q, \Sigma,\delta,q_0,Q_{acc},Q_{rej})$ where
\begin{itemize}
\item
$Q$ is a finite set of states, 
\item 
$\Sigma$ is an input alphabet,
\item $q_0 \in Q$ is a starting state, 
\item $Q_{acc} \subseteq Q$ is a set of accepting states,
\item $Q_{rej} \subseteq Q$ is a set of
rejecting states,
\item $\delta:Q \times \Gamma \times Q \rightarrow \mathds{C}_{[0,1]}$ is the transition function 
\end{itemize}
The state space of $M$ is the Hilbert space $\mathbb{C}^Q$ with basis $\{\ket{q} : q \in Q\}$.
\end{definition}

Let the left end marker $\#$ does not belong to $\Sigma$ and $\Gamma = \Sigma \cup \{\#\}$. 
Let  $Q_{halt}=Q_{acc}\cup Q_{rej}$ and $|Q_{halt}|=n$. 

A computation of $M$ on the input word $X=\#\sigma_1\sigma_2\dots\in \Gamma^{\omega}$ proceeds as follows: We require all $V_{\sigma_i}$ to be unitary.
The computation of $M$ starts in the superposition $\ket{q_0}$. A unitary transition corresponding to the current input symbol is performed. After every transition, $M$ measures its state with respect to the observable as follows: 
$$\bigoplus_{i=1}^{n}S_i \oplus S_{non} \; \text{where}\; S_i=span\{\ket{q_i}\}:q_i\in Q_{halt} \;\text{and}\;S_{non}=span\{\ket{q}:q\notin Q_{halt}\}$$ After each unitary transition, a measurement is performed. If the resulting state lies in some halting subspace \( S_i \), corresponding to \( q_i \in Q_{\text{halt}} \), then the computation halts. Otherwise, the state lies in $S_{\text{non}}$, and the computation proceeds to the next input symbol from the resulting (generally unnormalized) state in $S_{\text{non}}$.

The computation of $M$ on input $x=\sigma_1\sigma_2\sigma_3\cdots \in \Sigma^\omega$
starting from an arbitrary initial state $\psi_0 \in S_{\text{non}}$
is the sequence
$\psi_0,\ \psi_1,\ \psi'_1,\ \psi_2,\ \psi'_2,\ldots$
defined as follows.

First define the post-measurement state at step $0$ by
$\psi'_0 := P_{\text{non}}(\psi_0)$.
Since $\psi_0 \in S_{\text{non}}$, we have $\psi'_0 = \psi_0$. For every $j \ge 1$, define: $\psi_j := V_{\sigma_j}(\psi'_{j-1})$ (pre-measurement state), $\psi'_j := P_{\text{non}}(\psi_j)$ (post-measurement non-halting state).
We now define how an infinite word is recognized by an $\text{MMQBA}$ $M$ for a fixed cutpoint $p\in (0,1)$.

At any time during the computation $M$ is in the total state
$(\psi, acc, rej)$, where $\psi$ is a superposition
of non-halting states, $acc$ is the cumulative acceptance
probability, and $rej$ is the cumulative rejection probability.
For each $\sigma \in \Gamma$, the evolution on input $\sigma$ is given
by the operator $T_{\sigma}$ defined as
$$
T_{\sigma}(\psi, acc, rej)
\longrightarrow
\big(
P_{\text{non}}V_{\sigma}\psi,\;
acc + \|P_{acc}V_{\sigma}\psi\|^2,\;
rej + \|P_{rej}V_{\sigma}\psi\|^2
\big),
$$
where $P_{acc}$, $P_{rej}$, and $P_{\text{non}}$
are the projections onto the accepting, rejecting, and non-halting
subspaces, respectively, and $P_{\text{halt}} := P_{acc} + P_{rej}$ denote the projection onto the halting subspace.

Let $x = \#\sigma_1\sigma_2\dots \in \Gamma^{\omega}$ be an input word with the left end marker $\#$, where  $\Gamma = \Sigma \cup \{\#\}$, and let $x_j$ denote its prefix of length $j$ over $\Sigma$. The total state after reading $x_j$ is
$T_{\#x_j}(\ket{q_0},0,0)
=
(\ket{\psi'_j}, acc_j, rej_j)$
where $acc_j$ and $rej_j$ denote the cumulative
acceptance and rejection probabilities up to step $j$, respectively,
and $\ket{\psi'_j}$ is the post-measurement non-halting state after
applying $P_{\text{non}}$. The cumulative acceptance and rejection
probabilities are defined by
$
acc_j = \sum_{i=1}^j \|P_{acc}(\psi_i)\|^2$ and 
$rej_j = \sum_{i=1}^j \|P_{rej}(\psi_i)\|^2$.

The $\text{MMQBA}$ $M$ accepts $x$ with probability $p$ if the run of $M$ on $x$ satisfies all of the following:
\begin{description}
    \item[Büchi condition:] $\|P_{acc}(\psi_j)\|^2 > 0 \quad \text{for infinitely many } j$. That is, $M$ visits accepting halting states infinitely often.
    \item[Cutpoint acceptance:] $\lim_{j \to \infty} acc_j \;\ge\; p$. The cumulative acceptance probability in the run converges to the cutpoint $p$.
    \item[Cutpoint rejection:] $ \lim_{j \to \infty} rej_j \;<\; p$.
    The cumulative rejection probability in the run remains strictly below the cutpoint $p$.
\end{description}

The $\omega$-language recognized by $M$ with cutpoint $p$, denoted by $L(M,p)$, is the set of all
infinite words accepted by $M$ with probability at least $p$. Formally,
$L(M,p) := \{\, x \in \Sigma^\omega \mid \text{$M$ accepts $x$ with probability at least $p$} \,\}$.
\begin{remark}
Let $M$ be an MMQBA and
$x = x_1 x_2 x_3 \cdots \in \Sigma^\omega$ an input word.
Let $\psi_i$ denote the pre-measurement state after reading the $i$-th symbol. Then the cumulative acceptance probabilities satisfy ${acc}_{i+1} > {acc}_i \quad \text{iff} \quad P_{{acc}}(\psi_{i+1}) \neq 0.$
\end{remark}
We illustrate the model by the following example.

\begin{example}\label{example-fixed}
Let $L = a(a+b)^\omega$, $\Sigma=\{a,b\}$. \\
We define the MMQBA $M = (Q,\Sigma,\delta,q_0,Q_{acc},Q_{rej})$
where
$Q=\{q_0,q_1,q_2\},
Q_{acc}=\{q_1\},
Q_{rej}=\{q_2\}$, cumulative acceptance probability is $0.98$
and $q_0$ is the initial state.

We use the basis order $(q_0,q_1,q_2)$ and define the unitary matrices by giving their columns (the $j$-th column is $V_\sigma\ket{q_{j-1}}$ for the corresponding symbol). The matrices are:

\[
V_a =
\begin{pmatrix}
\frac{1}{\sqrt{3}} & \frac{\sqrt{2}}{\sqrt{3}} & 0 \\[6pt]
-\frac{\sqrt{2}}{\sqrt{3}} & \frac{1}{\sqrt{3}} & 0 \\[6pt]
0 & 0 & 1
\end{pmatrix},
\qquad
V_b =
\begin{pmatrix}
\frac{1}{3} & \frac{2}{3} & \frac{2}{3} \\[6pt]
\frac{2}{3} & \frac{1}{3} & -\frac{2}{3} \\[6pt]
\frac{2}{3} & -\frac{2}{3} & \frac{1}{3}
\end{pmatrix}.
\]
Let the input word $x$ be $aaabbb \dots$ and $\psi_0=\ket{q_0}$
\begin{itemize}
\item
$V_a |q_0\rangle = \tfrac{1}{\sqrt{3}}\ket{q_0} - \tfrac{\sqrt{2}}{\sqrt{3}}\ket{q_1}
=\psi_1$ where $x_1=a,\; acc_1=\tfrac{2}{3},\; rej_1=0$.

\item $V_{a}(\psi'_1) = \tfrac{1}{3}\ket{q_0} - \tfrac{\sqrt{2}}{3}\ket{q_1}=\psi_2$  
where $x_2=aa,\; acc_2=\tfrac{2}{3}+\tfrac{2}{9},\; rej_2=0,\; 
\psi'_1 = P_{\text{non}}(\psi_1) = \tfrac{1}{\sqrt{3}}\ket{q_0}$.

\item $V_{a}(\psi'_2) = \tfrac{1}{3\cdot\sqrt{3}}\ket{q_0} - \tfrac{\sqrt{2}}{3\cdot\sqrt{3}}\ket{q_1}=\psi_3$  
where $x_3=aaa,\; acc_3=\tfrac{2}{3}+\tfrac{2}{9}+\tfrac{2}{27},\; rej_3=0,\;
\psi'_2 = P_{\text{non}}(\psi_2)=\tfrac{1}{{3}}\ket{q_0}$.

\item $V_{b}(\psi'_3) = \tfrac{1}{9\cdot\sqrt{3}}\ket{q_0}+\tfrac{2}{9\cdot\sqrt{3}}\ket{q_1}+\tfrac{2}{9\cdot\sqrt{3}}\ket{q_2}=\psi_4$  
where $x_4=aaab,\; acc_4=\tfrac{2}{3}+\tfrac{2}{9}+\tfrac{2}{27}+\tfrac{4}{81\cdot 3},\; rej_4=\tfrac{4}{81\cdot 3},\; \psi'_3 = P_{\text{non}}(\psi_3)=\tfrac{1}{{3\sqrt{3}}}\ket{q_0}$.

\item $V_{b}(\psi'_4) = \tfrac{1}{27.\sqrt{3}}\ket{q_0}+\tfrac{2}{27.\sqrt{3}}\ket{q_1}+\tfrac{2}{27.\sqrt{3}}\ket{q_2}=\psi_5$  
where $x_4=aaabb,\; acc_5=\tfrac{2}{3}+\tfrac{2}{9}+\tfrac{2}{27}+\tfrac{4}{81\cdot 3}+\tfrac{4}{27\cdot 27\cdot 3},\; rej_5=\tfrac{4}{81\cdot 3}+\tfrac{4}{27\cdot 27\cdot 3},\; \psi'_4 = P_{\text{non}}(\psi_4)=\tfrac{1}{{27\sqrt{3}}}\ket{q_0}$.

\item In general, at step $j \geq 1$ the cumulative acceptance probability satisfies
$$
acc_j =
\begin{cases}
\displaystyle \frac{2}{3}, & j=1, \\[8pt]
\displaystyle \frac{2}{3} + \frac{2}{9} + \cdots + \frac{2}{3^j}, & 1 \le j \le 3, \\[10pt]
\displaystyle \frac{2}{3} + \frac{2}{9} + \frac{2}{27} \;+\;
\sum_{k=4}^j \frac{4}{3 \cdot 9^{\,k-2}}, & j \ge 4,
\end{cases}
$$ 
$$
\text{ and }rej_j =
\begin{cases}
0, & 1 \le j \le 3, \\[8pt]
\displaystyle \sum_{k=4}^j \frac{4}{3 \cdot 9^{\,k-2}}, & j \ge 4.
\end{cases}
$$
$M$ reads infinite prefixes of $x$ and reach the accepting and rejecting states $q_1,q_2$ infinitely many times respectively: Since this is a geometric series:
$$
\lim_{j \to \infty} acc_j = acc_{\infty}
= \frac{2}{3} + \frac{2}{9} + \frac{2}{27}
+ \cdots=\frac{\tfrac{4}{243}}{1 - \tfrac{1}{9}}
= \frac{26}{27} + \frac{1}{54}
= \frac{53}{54} \approx 0.98148
$$
$$
\text{ and }\lim_{j \to \infty} rej_j = rej_{\infty}
=\frac{1}{54} \approx 0.01852
$$
\end{itemize}
Let the input word $x$ be $bbb\cdots$. Now we illustrate rejection.
$V_b |q_0\rangle = \tfrac{1}{3}\ket{q_0}+\tfrac{2}{3}\ket{q_1}+\tfrac{2}{3}\ket{q_2}$, 
so $acc_1=\tfrac{4}{9}, rej_1=\tfrac{4}{9}, \psi_1=P_{\text{non}}(\psi)=\tfrac{1}{3}|q_0\rangle$.
This pattern repeats, giving $\lim_{j\to\infty} acc_j=\tfrac{1}{2},
\lim_{j\to\infty} rej_j=\tfrac{1}{2}$. Hence $b^\omega$ is rejected for any cutpoint $p>\tfrac12$.
\end{example}
We briefly explain the behavior of $M$. For any input $x \in L = a(a+b)^\omega$, 
the first symbol $a$ ensures that a positive amount of probability is transferred 
to the accepting subspace $S_{acc}$. Subsequent transitions maintain a bias towards 
acceptance, and the cumulative acceptance probability converges to a value 
approximately $0.98$, which is strictly greater than $\tfrac{1}{2}$.

On the other hand, for inputs $x \notin L$, this acceptance bias is not established, 
and the cumulative acceptance probability does not exceed the cutpoint, while 
rejection dominates. Hence, $M$ recognizes $L$ according to the definition of MMQBA.

This example demonstrates the evolution of an MMQBA and prepares the ground 
for structural properties studied in the next sections. 

\section{Structural Properties of MMQBA}
\label{structural}
\par This section establishes structural properties of measure-many quantum Büchi automata (MMQBA) that are used in subsequent language-theoretic results, focusing on the behavior of the non-halting subspace during infinite computations. We show that this subspace admits a decomposition into components with distinct long-term behaviors, enabling a characterization of how halting probability is generated and how acceptance probability accumulates along a run. Intuitively, this parallels the classical setting where non-accepting states can be divided into those from which acceptance remains possible and those forming trap states. The quantum analogue, formalized in the following lemma, decomposes the non-halting subspace into one component where the computation evolves without ever producing halting probability, and another from which probability mass gradually leaks into the halting subspace, thereby enabling acceptance or rejection.
\begin{lemma}[Non-halting subspace decomposition]\label{subspace decomposition}
Let $M$ be an MMQBA with state space $\mathbb{C}^Q$, and let
$S_{\text{non}}$ denote the non-halting subspace.
Then there exist subspaces $S^1, S^2 \subseteq S_{\text{non}}$
such that $S_{\text{non}} = S^1 \oplus S^2$ and the following properties hold. For any $\psi \in S_{\text{non}}$ and any infinite word
$x=\sigma_1\sigma_2\cdots \in \Sigma^\omega$, let
$\psi_0,\psi_1,\psi'_1,\psi_2,\psi'_2,\ldots$
denote the computation of $M$ on input $x$ starting from $\psi$.

\begin{enumerate}
\item For every $\psi \in S^1$, the computation never leaves
$S^1$, and no halting probability is ever produced. Formally,
$\forall k \ge 1,\ \psi'_k \in S^1$ and
$\sum_{k=1}^{\infty} \|P_{\text{halt}}(\psi_k)\|^2 = 0$.
\item For every $\psi \in S^2$, the non-halting component converges to
zero. Formally,
$\lim_{k\to\infty} \|\psi'_k\| = 0$.
\end{enumerate}
\end{lemma}
Proof is omitted due to space constraint.

Note that if the initial state were contained
entirely in $S^1$, then by Lemma~\ref{subspace decomposition} the
computation would remain inside $S^1$ and no halting probability would
ever be generated. Consequently the automaton would accept no word,
i.e., $L(M,p)=\emptyset$.
\begin{corollary}\label{corollary11}
Let $\psi = \phi_1 + \phi_2$, where $\phi_1 \in S^1$ and $\phi_2 \in S^2$.
Let the run of $M$ on $x = \sigma_1 \sigma_2 \sigma_3 \cdots$ be defined as in an MMQBA, and let $\psi_k$ denote the pre-measurement state at step $k$. Then, for all $k \ge 1$, $\| P_{\text{halt}}(\psi_k) \|^2
=
\| P_{\text{halt}}(V_{\sigma_k} P_{\text{non}} V_{\sigma_{k-1}} \cdots
P_{\text{non}} V_{\sigma_1}(\phi_2)) \|^2.$
\end{corollary}
Proof is omitted due to space constraint.

\begin{theorem}\label{convergence}
Let $x = \sigma_1 \sigma_2 \cdots \in \Sigma^\omega$ and $M$ be an MMQBA. If
    $\|P_{acc}(\psi_j)\|^2 > 0$ for infinitely many $j$ and
    $\sum_{j=1}^\infty \|P_{rej}(\psi_j)\|^2 = 0$,
then the cumulative acceptance probability converges to 1, that is $\lim_{j \to \infty} \text{acc}_j=1$.
\end{theorem}

\begin{proofsketch}

We show that the cumulative probability of accepting prefixes converges to 1. Let $\psi_0 = |q_0\rangle$ be the initial quantum state. We recall the pre-measurement state as $\psi_j$ and the
post-measurement non-halting state as $\psi'_j$. The halting probabilities
at step $j$ are defined by
$$
\alpha_j := \|P_{acc}(\psi_j)\|^2,
\;
\rho_j := \|P_{rej}(\psi_j)\|^2,\;
acc_j := \sum_{i=1}^j \alpha_i,
\;
rej_j := \sum_{i=1}^j \rho_i
$$
For the first step, since $\psi_0 \in S_{\text{non}}$, $\|\psi'_1\|^2 = \|\psi_0\|^2 - \alpha_1 - \rho_1.$ For every $j \ge 2$, the evolution satisfies the norm conservation identity
$\|\psi'_j\|^2=\|\psi'_{j-1}\|^2
- \alpha_j -\rho_j$. 
By induction on $j$, it follows that
$
\|\psi'_j\|^2
=
1
-
acc_j
-
rej_j$.

We are given: (i) \( \alpha_j > 0 \) for infinitely many \( j \in \mathbb{N} \) (i.e., accepting states are visited infinitely often), and (ii) \( \sum_{j=1}^\infty \rho_j = 0 \) (i.e., no amplitude is ever lost to rejecting states). From (ii) it follows that \( \rho_j = 0 \) for all \( j \), so \( \text{rej}_j = 0 \), and hence \( \|\psi'_j\|^2 = 1 - \text{acc}_j \).

The total acceptance probability is distributed across infinitely many steps. As the computation never halts (it resumes from the non-halting subspace at each step), this mass gradually accumulates in accepting states over time rather than collapsing at a finite point. Thus, the cumulative acceptance probability
$acc_j = \sum_{i=1}^j \alpha_i$ tends to $1$ as $j$  tends to $\infty$, i.e.,  $\lim_{j \to \infty} acc_j = 1$.
\end{proofsketch}
\begin{remark}
The converse of Theorem~\ref{convergence} does not hold.
That is, a run may accumulate total acceptance probability~$1$
without visiting accepting states infinitely often.
Indeed, consider the following MMQBA.
Let $\Sigma=\{a\}$ and define an automaton
$M=(Q,\Sigma,\delta,q_0,Q_{acc},Q_{rej})$
with $Q=\{q_0,q_1\}$,
$Q_{acc}=\{q_1\}$,
and $Q_{rej}=\emptyset$.
Define the unitary transition $V_a = \begin{bmatrix}
0 & 1 \\
1 & 0
\end{bmatrix}
$. On the input $x=a^\omega$, the automaton transitions as $V_a \ket{q_0} = \ket{q_1}$. Thus the automaton halts at the first step with acceptance probability~$1$, and the computation terminates immediately. Consequently, $
\lim_{j\to\infty} acc_j = 1$. However, the accepting state is visited only once. Hence the Büchi condition $\|P_{acc}(\psi_j)\|>0 \quad \text{for infinitely many } j$
is violated. This shows that convergence of cumulative acceptance
probability to~$1$ does not imply infinite visits to accepting states.
\end{remark}
\section{Limit Characterization and Non-recognizability of MMQBA}
\medskip
\noindent
In this section we investigate the relationship between
MMQFA and MMQBA recognizability via the limit operator.
We first establish a structural lemma concerning the behavior of
non-halting subspaces under a fixed transition,
which serves as the main technical tool for the subsequent arguments.
We then prove a limit-characterization theorem.
Finally, we introduce a class of languages whose limit
$\omega$-languages are not recognizable by MMQBA.
\begin{definition}
Let $M$ be an MMQBA with non-halting subspace $S_{\text{non}}$,
and let $\sigma \in \Sigma$. A subspace $S \subseteq S_{\text{non}}$
is called a \emph{$\sigma$-cycle subspace} of $M$ if $V_\sigma(S) \subseteq S$, that is,
$\forall \psi \in S,V_\sigma(\psi) \in S.$
\end{definition}
\begin{lemma}[No-entry into a $\sigma$-cycle subspace]\label{no}
Let $M$ be an MMQBA with basis state set $Q$ and
$S \subseteq S_{\text{non}}$ be a $\sigma$-cycle subspace.
For any basis vector $|r\rangle \in Q$ with 
$|r\rangle \notin S$, then
$P_S\!\big(V_\sigma |r\rangle\big)=0.$
\end{lemma}

The previous lemma formalizes a crucial structural phenomenon:
once a non-halting subspace becomes invariant under a transition
$V_\sigma$, no computation starting outside that subspace can ever
enter it. This restriction on the movement of amplitude between
subspaces will play a decisive role in proving non-recognizability
results later in this section. Proof is omitted due to space constraint.

\begin{definition}[Limit of a language]\cite{baier2008principles}
For $L \subseteq \Sigma^*$, define $\lim(L) = \{\, x \in \Sigma^\omega \mid
\text{$x$ has infinitely many prefixes in $L$} \,\}.$
\end{definition}
\begin{theorem}[limit-characterization]\label{limit-characterization}
Let $L' \subseteq \Sigma^\omega$.  
There exist an MMQBA  $M'$ and a cutpoint  $p' > \tfrac12$ such that $L' = L(M',p')
$ iff there exist an MMQFA $M$ and a cutpoint $p >\tfrac12$  such that $L' = \lim\!\bigl(L(M,p)\bigr)$
\end{theorem}

\begin{proofsketch}
We prove both directions.

\medskip

\noindent
\textbf{($\Rightarrow$)}
Let $M=(Q,\Sigma,\delta,q_0,Q_{acc},Q_{rej})$ be an
MMQFA with cutpoint $p>\tfrac12$, and let $L=L(M,p)$.
We construct an MMQBA
$M'=(Q',\Sigma,\delta',q'_0,Q'_{acc},Q'_{rej})$
with cutpoint $p'>\tfrac12$ such that $L(M',p')=\lim(L)$.

The construction keeps the same state space and accepting/rejecting
subspaces, i.e.,
$Q'=Q$, $Q'_{acc}=Q_{acc}$,
$Q'_{rej}=Q_{rej}$, and $q'_0=q_0$.
The difference is that an MMQBA has no right end-marker.

The construction distributes the effect of the end-marker
$V_{\$}$ over infinitely many steps. Fix $\varepsilon \in (0,1)$.
For each prefix $u$, if $u \in L$, then a fraction $\varepsilon$
of the current non-halting probability is transferred to the
halting subspaces, while a fraction $(1-\varepsilon)$ remains
in $S_{{non}}$. This redistribution preserves total
probability and the structure of superpositions, and hence
can be realized by unitary transitions.

If $x \in \lim(L)$, then $x$ has infinitely many prefixes
$u_j \in L$. Let $r_j$ denote the non-halting probability
after processing $u_j$, then
$$
r_j \le (1-\varepsilon)^j \longrightarrow 0.
$$
Thus the total halting probability tends to $1$.
Moreover, each such prefix contributes acceptance with bias
towards $S_{acc}$, and hence acceptance dominates in the limit.
Therefore, $x \in L(M',p')$ for some $p'>\tfrac12$.

\medskip

\noindent
\textbf{($\Leftarrow$)}
Let $M'=(Q,\Sigma,\delta',q_0,Q_{acc},Q_{rej})$ be an MMQBA
recognizing $\lim(L)$ with cutpoint $p'>\tfrac12$.
We construct an MMQFA
$M=(Q,\Sigma,\delta,q_0,Q_{acc},Q_{rej}),$
again preserving the state space and partitions. Fix $\varepsilon \in (0,1)$ such that $(1-\varepsilon)>\tfrac12$.

The transitions of $M$ are defined so that, after reading any
$u \in \Sigma^*$,
$$
\|P_{{non}}(\psi_u)\|^2 \ge (1-\varepsilon) > \tfrac12.
$$
Such transitions preserve total probability and the structure
of superpositions, and hence can be implemented by unitary operators.
Thus most of the probability remains in the non-halting subspace.

Upon reading the end-marker $\$$, the remaining non-halting
probability is entirely transferred to the halting subspaces,
so that the final acceptance-rejection distribution reflects
whether $u \in L$ or $u \notin L$. Since more than half of the
total probability resides in the non-halting component before
applying $V_{\$}$, this majority determines the final decision,
yielding a cutpoint $p>\tfrac12$.

\medskip

\noindent
This establishes the equivalence.
\end{proofsketch}
\begin{definition}[$\sigma-block$ forcing language]
Let $\Sigma$ be an alphabet and $\sigma \in \Sigma$. A language $L \subseteq \Sigma^*$ is called a $\sigma-block$ forcing language if every infinite word $x\in \lim(L)$ has the following form, such that,
$x=\sigma^{k_1}u\sigma\sigma^{k_2}u\sigma\sigma^{k_3}u\sigma\dots$ where $k_i \geq 1$ for all $i \geq 1$, $u \in \Sigma^{+}$, $u \notin \sigma^{+}$. ($u$ contains at least one symbol different from $\sigma$)
\end{definition}
\begin{example}
We present an example of $\sigma$-block forcing languages.
 and fix $\sigma=a$.  
Consider the language
$L=(a^+bca)^*$. Then $\lim(L)=(a^+bca)^\omega$.
Every infinite word in $\lim(L)$ has the form $a^{k_1}bca\,a^{k_2}bca\,a^{k_3}bca\cdots$ where $k_i\ge1$ for all $i\ge1$. Here $u=bc\in\Sigma^+$ and $u\notin a^+$.
Hence $L$ is an $a$-block forcing language.
\end{example}

The key property of a $\sigma$-block forcing language is that every $x \in \lim(L)$ contains infinitely many long $\sigma$-blocks separated by fixed non-$\sigma$ segments. In an MMQBA, repeated application of $V_\sigma$ over each block forces the finite-dimensional computation to eventually stabilize within a $V_\sigma$-invariant non-halting subspace. However, the block-forcing structure requires acceptance leakage at the end of every block, compelling the computation to exit any previously occupied invariant subspace. By the no-entry lemma~\ref{no}, re-entry is impossible, thus requiring infinitely many pairwise disjoint invariant subspaces—contradicting finite dimensionality.

\begin{theorem}[Non recognizability of $\sigma$-block forcing language]\label{non}
If $L \subseteq \Sigma^*$ is a $\sigma$-block forcing language then there exists no MMQBA $M'$ with cutpoint $p' > \frac{1}{2}$ such that $L(M',p')=\lim(L)$.
\end{theorem}
\begin{proofsketch} 
Assume for contradiction that an MMQBA $(Q,\Sigma,\delta,q_0,Q_{acc},Q_{rej})$ with cutpoint $p'>\tfrac12$ recognizes $\lim(L)$. Since $L$ is a $\sigma$-block forcing language, every word
$x$ in $\lim(L)$ has the form
$\sigma^{k_1}u\sigma\sigma^{k_2}u\sigma\sigma^{k_3}u\sigma\cdots$
where $k_i\ge1$, $u\in\Sigma^+$, and $u\notin\sigma^+$.
Hence the transition $V_\sigma$ is applied infinitely many times
during the run of $M'$ on $x$.

Let $S_{\text{non}}$ denote the non-halting subspace and
$Q_{\text{non}}=Q\setminus Q_{\text{halt}}$.
For each $\sigma$-block in the run, let $\psi_{i,m}$ denote the
non-halting state obtained after the $m$-th application of
$V_\sigma$ in the $i$-th block. Since infinitely many non-halting states appear during the run and
$Q_{\text{non}}$ is finite, some basis states must occur with
nonzero amplitude infinitely often. We define
\[
B_\infty :=
\{ q\in Q_{\text{non}} :
q \text{ has nonzero amplitude in infinitely many
states } \psi_{i,m}\} \mbox{ and }\]
\[
S^{(1)}=\operatorname{span}\{\,|q\rangle : q\in B_\infty\,\}
\subseteq S_{\text{non}} .
\]
Then $S^{(1)}\neq\{0\}$.
Moreover, if $q\in B_\infty$, then $q$ appears in infinitely many
states $\psi_{i,m}$ and
$\psi_{i,m+1}=P_{\text{non}}(V_\sigma\psi_{i,m})$.
If $V_\sigma|q\rangle$ had a component on some basis state
$r\notin B_\infty$, then $r$ would also appear in infinitely many
subsequent non-halting states, contradicting the definition of
$B_\infty$. Hence $V_\sigma(S^{(1)})\subseteq S^{(1)},$so $S^{(1)}$ is a nonzero $V_\sigma$-cycle subspace.

Next we show that vectors in $S^{(1)}$ produce no acceptance under
$V_\sigma$. If some $\psi\in S^{(1)}$ satisfied
$\|P_{acc}(V_\sigma\psi)\|^2>0$, then along the infinite
input $\sigma^\omega$ acceptance leakage would occur infinitely
many times, forcing the cumulative acceptance probability to
exceed the cutpoint. This would imply
$\sigma^\omega\in\lim(L)$, contradicting the definition of a
$\sigma$-block forcing language. Hence
$\|P_{acc}(V_\sigma\psi)\|^2=0
\quad\text{for all }\psi\in S^{(1)}.$

Now consider the configuration after reading the prefix
$\sigma^{k_1}u$ of $x$. If the computation halted at this point,
the remaining infinite input could not be processed, contradicting
the semantics of MMQBA. Hence the computation continues from a
non-halting state.

Consider the next symbol $\sigma$ in the prefix
$\sigma^{k_1}u\sigma$. Since this prefix occurs in some word of
$L$, the acceptance–rejection bias satisfies $
\|P_{acc}(V_\sigma\psi_1)\|^2>
\|P_{rej}(V_\sigma\psi_1)\|^2$
and in particular a positive accepting component is produced.
After measurement the non-halting residue
$\psi_1'=P_{\text{non}}(V_\sigma\psi_1)$ is nonzero.
Since vectors of $S^{(1)}$ produce no accepting component under
$V_\sigma$, we must have $\psi_1'\notin S^{(1)}$.

Because $S^{(1)}$ is $V_\sigma$-invariant, once the computation
leaves $S^{(1)}$ it can never re-enter it. During the next
$\sigma$-block $\sigma^{k_2}$ the same argument produces another
nonzero $V_\sigma$-cycle subspace
$S^{(2)}\subseteq S_{\text{non}}$ with
$S^{(2)}\cap S^{(1)}=\{0\}$.
Repeating this reasoning for each block $\sigma^{k_m}$ yields
infinitely many pairwise disjoint nonzero subspaces $S^{(1)},S^{(2)},S^{(3)},\ldots$ satisfying $V_\sigma(S^{(m)})\subseteq S^{(m)}.$
This contradicts the finite dimension of $S_{\text{non}}$.
Hence no MMQBA with cutpoint greater than $\tfrac12$
can recognize $\lim(L)$. 
\end{proofsketch}
\begin{corollary}\label{cor1}
If $L \subseteq \Sigma^*$ is a $\sigma$-block forcing language, then
there exists no MMQFA $M$ with cutpoint $p > \tfrac12$ such that
$L(M,p)=L$.
\end{corollary}
Proof is omitted due to space constraint.
\begin{corollary}\label{cor2}
If $L\subseteq\Sigma^*$ is a finite language, then
there exist an MMQFA $M$ and an MMQBA $M'$ such that $L(M,1)=L$ and $L(M',1)=\lim(L)$.
\end{corollary}
Proof is omitted due to space constraint.
\section{Closure Properties of MMQBA}
We present closure properties of MMQBA in this section. 
The class of languages recognized by MMQBA is not closed under union, intersection and complement. For completeness.
\begin{theorem}\label{union}
The class of languages recognized by MMQBA is not closed under 
union.
\end{theorem}
\begin{proofsketch}
Full construction is presented in Theorem \hyperlink{union}{4} of Appendix.

Let $L_1=(ab)^\omega$ and $L_2=(aab)^\omega$, both recognizable by
MMQBA, these constructions are straight forward. 
Assume for contradiction that there exists an MMQBA $M$ with
cutpoint $p>\tfrac12$ recognizing $L=L_1\cup L_2$.

By Lemma~\ref{subspace decomposition}, the non-halting subspace
decomposes as $S_{\text{non}}=S^1\oplus S^2$, and the initial state
must have a nonzero component in $S^2$. Consider the state $\psi_{aa}$
reached after reading the prefix $aa$ from the initial state. 

For the run on $x_1=(ab)^\omega$, the prefix $aa$ does not occur. If
$\psi_{aa}$ had a nonzero component in $S^2$, then continuing the
computation from $\psi_{aa}$ along $(ab)^\omega$ would again produce
infinitely many accepting projections, causing words with
prefix $aa$ to be accepted. Hence $\psi_{aa}\in S^1$.

However, in the word $x_2=(aab)^\omega$ the prefix $aa$ occurs
infinitely often, and acceptance requires that the computation after
each occurrence continues to generate accepting amplitude. Thus the
state after reading $aa$ must lie in $S^2$, i.e., $\psi_{aa}\in S^2$.

This contradicts $S^1\cap S^2=\{0\}$. Therefore no MMQBA can recognize
$L_1\cup L_2$, and the class of MMQBA-recognizable $\omega$-languages
is not closed under union.
\end{proofsketch}
\begin{theorem}\label{intersection}
The class of languages recognized by MMQBA is not closed under intersection.
\end{theorem}
\begin{proofsketch}
Let $\Sigma=\{a,b\}$ and define
$
L_a=\lim\{\,w \mid w \text{ has an odd number of $a$'s}\,\}$
 and 
$L_b=\lim\{\,w \mid w \text{ has an odd number of $b$'s}\,\}.
$ Both languages are recognizable by MMQBA, these are trivial constructions.

Let $L=L_a\cap L_b$. Suppose, for contradiction, that there exists
an MMQBA $M$ with cutpoint $p>\tfrac12$ such that $L(M,p)=L$.
Then for all $k,\ell\ge1$ we have $(a^k b^\ell)^\omega\in L(M,p)$,
while $a^{2t}b^\omega\notin L(M,p)$ for every $t\ge1$.

Analyzing the run of $M$ on $(a^k b^\ell)^\omega$, the
non-halting space decomposes as $S_{\text{non}}=S^1\oplus S^2$.
States reached after the $a$-segments lie in $S^1$, while accepting
amplitude is generated from $S^2$ during the $b$-segments. The run
therefore repeatedly alternates ($S^1 \rightarrow S^2 \rightarrow S^1\dots$)
along the periodic input.

Using this structure one can construct $t\ge1$ such that the run of
$M$ on $a^{2t}b^\omega$ still accumulates acceptance probability
above the cutpoint and satisfies the Büchi condition. Hence
$a^{2t}b^\omega\in L(M,p)$, contradicting the definition of $L$.
Therefore the class of languages recognized by MMQBA is not closed
under intersection.
\end{proofsketch}
\begin{theorem}\label{complement}
The class of languages recognized by MMQBAs with cutpoint $p$ 
is not closed under complement.
\end{theorem}
\begin{proofsketch}
Assume for contradiction that the class of languages recognized by
MMQBA is closed under complement. Let $L_a = \lim\{\, w \in \Sigma^* \mid w \text{ has an odd number of } a\text{'s} \,\},$ 
and suppose there exists an MMQBA $M$ recognizing $\overline{L_a}$. Using the acceptance condition of $M$, we construct an infinite word $x = aab^{k_1}aab^{k_2}aab^{k_3}\cdots$ such that after reading each block $aab^{k_i}$ the run of $M$ reaches
an accepting configuration. Consequently the run of $M$ on $x$
contains infinitely many accepting visits and satisfies the MMQBA
acceptance conditions, so $M$ accepts $x$. However, the prefixes ending at the first $a$ of each block contain
$1,3,5,\ldots$ occurrences of $a$. Hence $x$ has infinitely many
prefixes with an odd number of $a$'s, implying $x \in L_a$.
Thus $M$ accepts a word belonging to $L_a$, contradicting the
assumption that $M$ recognizes $\overline{L_a}$. Therefore the class of languages recognized by MMQBA is not closed
under complement.
\end{proofsketch}
\section{Conclusion}
In this work, we introduced Measure-Many Quantum Büchi Automata (MMQBA) and
developed a formal semantic framework for their operation on infinite words.
The model combines Büchi recurrence with quantitative cutpoint conditions,
capturing the cumulative effect of repeated projective measurements during an
infinite computation. We established a language-theoretic characterization
showing that MMQBA-recognizable $\omega$-languages are precisely of the form
$\lim(L(M,p))$ for an MMQFA $M$. As part of the analysis, we developed structural
results describing the behavior of the non-halting subspace under repeated
applications of unitary transitions. Using these ideas, we proved a
non-recognizability theorem showing that certain $\omega$-languages cannot be
recognized by MMQBA. In addition, we showed that the class of
MMQBA-recognizable languages is not closed under union, intersection, complement.

We also presented two constructions of MMQBA recognizing the same $\omega$-language $a(a+b)^\omega$ using different Hilbert-space dimensions. This observation suggests natural questions regarding the minimization of MMQBA representations and the further development of MMQBA theory in analogy with classical Büchi automata. Understanding how the dimension of the underlying Hilbert space affects recognizability, developing systematic techniques for reducing the dimension of equivalent automata, investigating algorithmic decision problems such as emptiness, universality, and equivalence for MMQBA-recognizable $\omega$-languages, and exploring possible applications of MMQBA in modeling recurring behaviors in infinite quantum computations and quantum games remain interesting directions for future research.
\bibliographystyle{splncs04}
\bibliography{bibliography}
\appendix
\section*{Appendix: Additional Example and Technical Proofs}
\phantomsection

\begin{example}\label{example-fixed1}
Let $L = a(a+b)^\omega$, $\Sigma=\{a,b\}$. \\
We define the MMQBA $M = (Q,\Sigma,\delta,q_0,Q_{acc},Q_{rej})$
where
$Q=\{q_0,q_1,q_2,q_3\},
Q_{acc}=\{q_2\},
Q_{rej}=\{q_3\}$, cumulative acceptance probability is $0.88$
and $q_0$ is the initial state. The matrices are:
\[
V_a =
\begin{pmatrix}
0 & 0 & 0 & 1 \\[6pt]
\frac{2}{3} & \frac{1}{\sqrt{2}} & \frac{1}{3\cdot \sqrt{2}} & 0 \\[8pt]
\frac{2}{3} & -\frac{1}{\sqrt{2}} & \frac{1}{3\cdot \sqrt{2}} & 0 \\[8pt]
\frac{1}{3} & 0 & -\frac{4}{3\cdot \sqrt{2}} & 0
\end{pmatrix},
\qquad
V_b =
\begin{pmatrix}
0 & 0 & 0 & 1 \\[6pt]
0 & \frac{1}{\sqrt{2}} & -\frac{1}{\sqrt{2}} & 0 \\[6pt]
0 & \frac{1}{\sqrt{2}} & \frac{1}{\sqrt{2}} & 0 \\[6pt]
1 & 0 & 0 & 0
\end{pmatrix}.
\]
Let the input word $x$ be $aaabbb \dots$
\begin{itemize}
\item
$V_a |q_0\rangle = \frac{2}{3}\ket{q_1} + \frac{2}{3}\ket{q_2} + \frac{1}{3}\ket{q_3}=\psi_1$ where $x_1=a, acc_1=\frac{4}{9}, rej_1=\frac{1}{9}$
\item $V_{a}(\ket{\psi'_1}) = \frac{2}{3.\sqrt{2}}|q_1\rangle - \frac{2}{3.\sqrt{2}}|q_2\rangle $ where $x_2=aa,acc_2=\frac{4}{9}+\frac{4}{18}, rej_2=\frac{1}{9}, \psi'_1 = P_{non}(\psi_1)=\frac{2}{3} |q_1\rangle$

\item $V_{a}(\ket{\psi'_2}) = \frac{2}{3\cdot \sqrt{2}\cdot \sqrt{2}}|q_1\rangle - \frac{2}{3.\sqrt{2}\sqrt{2}}|q_2\rangle$ where $x_3=aaa,acc_3=\frac{4}{9} + \frac{4}{18} + \frac{4}{36}, rej_3=\frac{1}{9}, \psi'_2 = P_{non}(\psi_2)=\frac{2}{3\cdot \sqrt{2}} |q_1\rangle$
\item $V_{b}(\ket{\psi'_3}) = \frac{2}{3\cdot 2\cdot \sqrt{2}}|q_1\rangle + \frac{2}{3\cdot 2\cdot \sqrt{2}}|q_2\rangle$ where $x_4=aaab,acc_4=\frac{4}{9} + \frac{4}{18} + \frac{4}{36} + \frac{4}{72}, rej_4=\frac{1}{9}, \psi'_3 = P_{non}(\psi_3)=\frac{2}{3\cdot \sqrt{2}\sqrt{2}} |q_1\rangle$

\item In general, at step \( j \geq 1 \) we obtain the cumulative acceptance probability in the following way:
$$
acc_j = \sum_{k=1}^{j} \frac{4}{9 \cdot 2^{k-1}}
$$
\end{itemize}
$M$ will read infinite prefixes of $x$ and reach the final state $q_2$ infinitely many times: Since this is a geometric series:
\[
\lim_{j \to \infty} acc_j = acc_{\infty}
= \frac{4}{9} + \frac{4}{18}
+ \frac{4}{36} + \frac{4}{72} + \cdots 
=\frac{\tfrac{4}{9}}{1 - \tfrac{1}{2}}
=\tfrac{8}{9} \approx 0.88
\]
\[
\lim_{j \to \infty} rej_j = rej_{\infty}
= \frac{1}{9} =\tfrac{1}{9} \approx 0.11
\]
$M$ recognizes the input according to the definition of MMQBA.

Let the input word $x$ be $bbb \dots$, $M$ reads $bbb \dots$. Here, after reading first $b$, $M$ gets into a superposition of $q_3$ where $rej_1 = 1$ and the computation gets terminated. Now according to the definition of $MMQBA$, the input word $x$ will be rejected.
\end{example}
In the following three examples (Example 4,5,6), we provide the transition matrices defining the automaton for the language. The state transition details and probability calculations can be done similarly.
\begin{example}\label{ex3}
Let $L=(ab)^\omega$ where $\Sigma=\{a,b\}$. We define an MMQBA
$M=(Q,\Sigma,\delta,q_0,Q_{acc},Q_{rej})$ where
$Q=\{q_0,q_1,q_{acc},q_{rej}\}$,
$Q_{acc}=\{q_{acc}\}$,
$Q_{rej}=\{q_{rej}\}$,
cumulative acceptance probability is $1$, and $q_0$ is the initial state. The matrices are
\[
V_a=
\begin{pmatrix}
0&0&0&1\\
1&0&0&0\\
0&0&1&0\\
0&1&0&0
\end{pmatrix},
\qquad
V_b=
\begin{pmatrix}
0&\tfrac35&\tfrac45&0\\
0&0&0&1\\
0&\tfrac45&-\tfrac35&0\\
1&0&0&0
\end{pmatrix}.
\]
\end{example}
\begin{example}\label{ex5}
Let $L=(aab)^\omega$ where $\Sigma=\{a,b\}$. We define an MMQBA
$M=(Q,\Sigma,\delta,q_0,Q_{acc},Q_{rej})$ where
$Q=\{q_0,q_1,q_2,q_{acc},q_4,q_5\}$,
$Q_{acc}=\{q_{acc}\}$,
$Q_{rej}=\{q_4,q_5\}$,
cumulative acceptance probability is $1$, and $q_0$ is the initial state. The matrices are

\[
V_a=
\begin{pmatrix}
0&0&0&1&0&0\\
1&0&0&0&0&0\\
0&1&0&0&0&0\\
0&0&0&0&1&0\\
0&0&1&0&0&0\\
0&0&0&0&0&1
\end{pmatrix},
\qquad
V_b=
\begin{pmatrix}
0&0&\tfrac45&\tfrac35&0&0\\
0&0&0&0&1&0\\
0&0&0&0&0&1\\
0&0&\tfrac35&-\tfrac45&0&0\\
1&0&0&0&0&0\\
0&1&0&0&0&0
\end{pmatrix}.
\]
\end{example}
\begin{example}\label{ex4}
Let $L_a=\lim\{\,w\in\{a,b\}^* \mid w \text{ has an odd number of } a\text{'s}\,\}$ where $\Sigma=\{a,b\}$. We define $MMQBA$ $M=(Q,\Sigma,\delta,q_0,Q_{acc},Q_{rej})
$ where $Q=\{q_0,q_1,q_{acc},q_{rej}\},Q_{acc}=\{q_{acc}\},
Q_{rej}=\{q_{rej}\}$, cumulative acceptance probability is $1$ and $q_0$ is the initial state. The matrices are:
\[
V_a=
\begin{pmatrix}
0 & 1 & 0 & 0\\
\frac{\sqrt{3}}{2} & 0 & -\frac{1}{2} & 0\\
\frac{1}{2} & 0 & \frac{\sqrt{3}}{2} & 0\\
0 & 0 & 0 & 1
\end{pmatrix},
\qquad
V_b=
\begin{pmatrix}
1 & 0 & 0 & 0\\
0 & 1 & 0 & 0\\
0 & 0 & 1 & 0\\
0 & 0 & 0 & 1
\end{pmatrix}.
\]

\end{example}
Now we provide the detailed proof omitted from the main body of the paper.\\
\hypertarget{subspace decomposition}
\noindent\textbf{Lemma~\ref{subspace decomposition}.}
\textit{Let $M$ be an MMQBA with state space $\mathbb{C}^Q$, and let
$S_{\text{non}}$ denote the non-halting subspace.
Then there exist subspaces $S^1, S^2 \subseteq S_{\text{non}}$
such that $S_{\text{non}} = S^1 \oplus S^2$ and the following properties hold. For any $\psi \in S_{\text{non}}$ and any infinite word
$x=\sigma_1\sigma_2\cdots \in \Sigma^\omega$, let
$\psi_0,\psi_1,\psi'_1,\psi_2,\psi'_2,\ldots$
denote the computation of $M$ on input $x$ starting from $\psi$.
\begin{enumerate}
\item For every $\psi \in S^1$, the computation never leaves
$S^1$, and no halting probability is ever produced. Formally,
$\forall k \ge 1,\ \psi'_k \in S^1$ and
$\sum_{k=1}^{\infty} \|P_{\text{halt}}(\psi_k)\|^2 = 0$.
\item For every $\psi \in S^2$, the non-halting component converges to
zero. Formally,
$\lim_{k\to\infty} \|\psi'_k\| = 0$.
\end{enumerate}}
\begin{proof}
Let $V_\sigma$ be the unitary transition operator for input symbol
$\sigma \in \Sigma$, and define the projected transition operator as
$V'_\sigma := P_{\text{non}} V_\sigma$. Let $\psi \in S_{\text{non}}$ be the initial state and define the run as follows:
\[
\psi_0 := \psi,
\qquad
\psi_1 := V_{\sigma_1}\psi_0,
\qquad
\psi'_1 := P_{\text{non}}(\psi_1).
\]
For every $k \ge 2$, define recursively:
\[
\psi_k := V_{\sigma_k}\psi'_{k-1},
\qquad
\psi'_k := P_{\text{non}}(\psi_k).
\]
The state space is finite-dimensional. Since the internal state set $Q$ is finite, the Hilbert space $\mathcal{H} := \mathbb{C}^Q$ is finite-dimensional. Therefore, any descending chain of subspaces in $\mathcal{H}$ must stabilize after finitely many steps \cite{Axler}. We define a descending sequence of subspaces. Let $\mathcal{W}_0 := S_{\text{non}}$. For each $i \ge 1$, define:
$$
\mathcal{W}_i := \left\{ \psi \in \mathcal{W}_{i-1} \;\middle|\; \forall \sigma \in \Sigma,\ V_\sigma' \psi \in \mathcal{W}_{i-1} \text{ and } \|P_{\text{halt}}(V_\sigma (\psi))\|^2 = 0 \right\}
$$
Thus $\mathcal{W}_i \subseteq \mathcal{W}_{i-1} \subseteq S_{\text{non}}$.
Intuitively, $\mathcal{W}_i$ consists of all vectors in $\mathcal{W}_{i-1}$
that remain in $\mathcal{W}_{i-1}$ under all projected transitions
and produce no halting amplitude at any step. This yields a descending chain:
$$
\mathcal{W}_0 \supseteq \mathcal{W}_1 \supseteq \mathcal{W}_2 \supseteq \cdots
$$
which must stabilize at some $i_0 \ge 0$, i.e.,
$\mathcal{W}_{i_0} = \mathcal{W}_{i_0+1} = \cdots$ 

We first verify that each $\mathcal{W}_i$ is a subspace of $\mathcal{H}$.

\smallskip
\noindent
Let $0 \in \mathcal{W}_{i-1}$. Then for all $\sigma \in \Sigma$,
$
V_\sigma' 0 = 0 \in \mathcal{W}_{i-1}, \quad P_{\text{halt}} V_\sigma 0 = 0.
$. Thus $0 \in \mathcal{W}_i$. Let $\psi_1, \psi_2 \in \mathcal{W}_i$, and $\alpha \in \mathbb{C}$. Since $\mathcal{W}_{i-1}$ is a subspace $
\psi_1 + \psi_2 \in \mathcal{W}_{i-1},\alpha \psi_1 \in \mathcal{W}_{i-1}$. Moreover, for all $\sigma \in \Sigma$, 
$$V_\sigma' (\psi_1 + \psi_2) = V_\sigma' \psi_1 + V_\sigma' \psi_2 \in \mathcal{W}_{i-1}, \|P_{\text{halt}}(V_\sigma (\psi_1 + \psi_2))\|^2 = 0$$
and for any scalar $\alpha \in \mathbb{C}$, 
$$
V_\sigma' (\alpha \psi_1) = \alpha V_\sigma' \psi_1 \in \mathcal{W}_{i-1}, \|P_{\text{halt}}(V_\sigma (\alpha \psi_1))\|^2 = 0.
$$

Thus $\psi_1 + \psi_2 \in \mathcal{W}_i$ and $\alpha \psi_1 \in \mathcal{W}_i$, and therefore $\mathcal{W}_i$ is a subspace.

We now define
$
S^1 := \mathcal{W}_{i_0},
\;
S^2 := S_{\text{non}} \cap (S^1)^\perp .
$

Since the descending chain $\mathcal{W}_0 \supseteq \mathcal{W}_1 \supseteq \cdots$
stabilizes at $\mathcal{W}_{i_0}$, the subspace $S^1$ is the \emph{maximal}
subspace of $S_{\text{non}}$ such that every run starting in $S^1$
remains in $S^1$ under all projected transitions and produces no halting
amplitude.

The subspace $S^2$ is the orthogonal complement of $S^1$ inside
$S_{\text{non}}$. By the orthogonal decomposition theorem \cite{Axler}, we obtain
$
S_{\text{non}} = S^1 \oplus S^2 .
$

We show Property~1, that is, no halting from $S^1$.
We prove by induction on $k \ge 1$ that
$$
\psi'_k \in S^1
\quad \text{and} \quad
\|P_{\text{halt}}(V_{\sigma_{k+1}} \psi'_k)\|^2 = 0 .
$$
\emph{Base case} ($k = 1$):  
By assumption, $\psi_0 \in S^1 = \mathcal{W}_{i_0}$.
Since $\psi_1 = V_{\sigma_1}\psi_0$ and $\psi'_1 = P_{\text{non}}(\psi_1)$,
and because $\psi_0 \in \mathcal{W}_{i_0}$, by definition of $\mathcal{W}_{i_0}$ we have
$$
\psi'_1 = V'_{\sigma_1}\psi_0 \in \mathcal{W}_{i_0} = S^1,
\quad
\|P_{\text{halt}}(\psi_1)\|^2 = 0 .
$$
\emph{Inductive step}:  
Assume $\psi'_k \in S^1 = \mathcal{W}_{i_0}$.
By definition of $\mathcal{W}_{i_0}$, we have
$$
V'_{\sigma_{k+1}} \psi'_k \in \mathcal{W}_{i_0}
\quad \text{and} \quad
\|P_{\text{halt}}(V_{\sigma_{k+1}}\psi'_k)\|^2 = 0 .
$$
Hence, $\psi'_{k+1} = V'_{\sigma_{k+1}} \psi'_k \in S^1$. By induction, we conclude that
$
\forall k \ge 1,\quad
\psi'_k \in S^1
\ \text{and}\
\|P_{\text{halt}}(V_{\sigma_{k+1}}\psi'_k)\|^2 = 0 .
$. Therefore, the cumulative halting probability is
$
\sum_{k=1}^\infty \|P_{\text{halt}}(\psi_k)\|^2 = 0 .
$

We show Property~2: norm vanishing from $S^2$.
Let $\psi \in S^2 = S_{\text{non}} \cap (S^1)^\perp$, and let
$x = \sigma_1 \sigma_2 \cdots \in \Sigma^\omega$ be any infinite input word. Consider the run of $M$ on $x$ starting from $\psi_0=\psi$, with
pre-measurement states $\psi_k$ and post-measurement states $\psi'_k$
defined as in the MMQBA semantics.

We claim that
$
\lim_{k \to \infty} \|\psi'_k\|
=
\lim_{k \to \infty} \|P_{\text{non}}(\psi_k)\|
= 0 .
$

We proceed by contradiction.
Suppose there exists $\varepsilon > 0$ and an infinite subsequence
$\{k_j\}$ such that $
\|\psi'_{k_j}\| \ge \varepsilon
\quad \text{for all } j.$
By the definition of MMQBA computation and using unitarity of
$V_{\sigma_k}$, the measurement step satisfies the following norm identities.

$$
\text{ For }
k = 1:
\|\psi'_1\|^2
=
\|\psi_0\|^2
-
\|P_{\text{halt}}(\psi_1)\|^2 .
$$
$$\text{ For every }k \ge 2:
\|\psi'_k\|^2
=
\|V_{\sigma_k}\psi'_{k-1}\|^2
-
\|P_{\text{halt}}(V_{\sigma_k}\psi'_{k-1})\|^2
=
\|\psi'_{k-1}\|^2
-
\|P_{\text{halt}}(\psi_k)\|^2 .
$$
Summing the norm identities from $k=1$ to $k=N$, we obtain
\begin{align*}
\|\psi'_1\|^2
&= \|\psi_0\|^2 - \|P_{\text{halt}}(\psi_1)\|^2,\\
\|\psi'_2\|^2
&= \|\psi'_1\|^2 - \|P_{\text{halt}}(\psi_2)\|^2,\\
\|\psi'_3\|^2
&= \|\psi'_2\|^2 - \|P_{\text{halt}}(\psi_3)\|^2,\\
&\ \vdots \\
\|\psi'_N\|^2
&= \|\psi'_{N-1}\|^2 - \|P_{\text{halt}}(\psi_N)\|^2 .
\end{align*}

Adding all these equalities, all intermediate terms cancel telescopically, yielding
$
\|\psi'_N\|^2
=
\|\psi_0\|^2
-
\sum_{i=1}^N \|P_{\text{halt}}(\psi_i)\|^2 .
$ Rearranging, we obtain
$
\sum_{i=1}^N \|P_{\text{halt}}(\psi_i)\|^2
=
\|\psi_0\|^2
-
\|\psi'_N\|^2 .
$ Taking the limit as $N \to \infty$, we get
$$
\sum_{i=1}^\infty \|P_{\text{halt}}(\psi_i)\|^2
=
\|\psi_0\|^2
-
\lim_{k\to\infty}\|\psi'_k\|^2 .
$$

If there exists $\varepsilon>0$ and an infinite subsequence $\{k_j\}$ such that
$\|\psi'_{k_j}\|\ge \varepsilon$, then
$$
\lim_{k\to\infty}\|\psi'_k\|^2 \ge \varepsilon^2 > 0,
\text{ and hence }
\sum_{i=1}^\infty \|P_{\text{halt}}(\psi_i)\|^2
<
\|\psi_0\|^2 .
$$

This contradicts the defining property of $S^2$, which requires that every
nonzero vector in $S^2$ contributes nontrivially to the halting amplitude over
time. Therefore,
$
\lim_{k\to\infty}\|\psi'_k\| = 0 .
$
\end{proof}

\noindent\textbf{Corollary~\ref{corollary11}.}
\textit{
Let $\psi = \phi_1 + \phi_2$, where $\phi_1 \in S^1$ and $\phi_2 \in S^2$.
Let the run of $M$ on $x = \sigma_1 \sigma_2 \sigma_3 \cdots$ be defined as in an MMQBA, and let $\psi_k$ denote the pre-measurement state at step $k$. Then, for all $k \ge 1$, $\| P_{\text{halt}}(\psi_k) \|^2
=
\| P_{\text{halt}}(V_{\sigma_k} P_{\text{non}} V_{\sigma_{k-1}} \cdots
P_{\text{non}} V_{\sigma_1}(\phi_2)) \|^2.$
}
\begin{proof}
By linearity of the unitary operators $V_{\sigma_i}$ and the projection
$P_{\text{halt}}$, we have
\[
P_{\text{halt}}(\psi_k)
=
P_{\text{halt}}(V_{\sigma_k} P_{\text{non}} \cdots V_{\sigma_1}(\phi_1))
+
P_{\text{halt}}(V_{\sigma_k} P_{\text{non}} \cdots V_{\sigma_1}(\phi_2)).
\]

By Lemma~\ref{subspace decomposition}, for every $k \ge 1$, the component originating from
$\phi_1 \in S^1$ never contributes to halting probability, and hence $\|P_{\text{halt}}(V_{\sigma_k} P_{\text{non}} \cdots V_{\sigma_1}(\phi_1))\|^2 = 0$.

Therefore, 
$$\| P_{\text{halt}}(\psi_k) \|^2
=
\| P_{\text{halt}}(V_{\sigma_k} P_{\text{non}} V_{\sigma_{k-1}} \cdots
P_{\text{non}} V_{\sigma_1}(\phi_2)) \|^2 .
$$
as claimed.
\end{proof}

\hypertarget{convergence}
\noindent\textbf{Theorem~\ref{convergence}.}
\textit{Let $x = \sigma_1 \sigma_2 \cdots \in \Sigma^\omega$ and $M$ be an MMQBA. If
    $\|P_{acc}(\psi_j)\|^2 > 0$ for infinitely many $j$ and
    $\sum_{j=1}^\infty \|P_{rej}(\psi_j)\|^2 = 0$,
then the cumulative acceptance probability converges to 1, that is $\lim_{j \to \infty} \text{acc}_j=1$}
\begin{proof}
We show that the cumulative probability of accepting prefixes converges to 1, that is $\lim_{j \to \infty} \text{acc}_j = 1.$
Let $\psi_0 = |q_0\rangle$ be the initial quantum state. 

We recall the intermediate (pre-measurement) state as $\psi_j$, and the
post-measurement non-halting state as $\psi'_j$. The halting probabilities
at step $j$ are defined by
$
\alpha_j := \|P_{acc}(\psi_j)\|^2,
\;
\rho_j := \|P_{rej}(\psi_j)\|^2,
$
and the cumulative probabilities by
$
acc_j := \sum_{i=1}^j \alpha_i,
\;
rej_j := \sum_{i=1}^j \rho_i.
$ For the first step, since $\psi_0 \in S_{\text{non}}$,
\[
\|\psi'_1\|^2 = \|\psi_0\|^2 - \alpha_1 - \rho_1 .
\]

For every $j \ge 2$, the evolution satisfies the norm conservation identity
\[
\|\psi'_j\|^2
=
\|\psi'_{j-1}\|^2
-
\alpha_j
-
\rho_j .
\]

By induction on $j$, it follows that
\[
\|\psi'_j\|^2
=
1
-
acc_j
-
rej_j .
\]

We are given: (i) \( \alpha_j > 0 \) for infinitely many \( j \in \mathbb{N} \) (i.e., accepting states are visited infinitely often), and (ii) \( \sum_{j=1}^\infty \rho_j = 0 \) (i.e., no amplitude is ever lost to rejecting states). From (ii) it follows that \( \rho_j = 0 \) for all \( j \), so \( \text{rej}_j = 0 \), and hence \( \|\psi'_j\|^2 = 1 - \text{acc}_j \).

The total acceptance probability is distributed across infinitely many steps. As the computation never halts (it resumes from the non-halting subspace at each step), this mass gradually accumulates in accepting states over time rather than collapsing at a finite point. Thus, the cumulative acceptance
$$
acc_j = \sum_{i=1}^j \alpha_i \to 1\text{ as } j \to \infty,\text{ i.e., } \lim_{j \to \infty} acc_j = 1, \text{ as required}.
$$

We now argue probabilistically. Let \( E_i \) be the event that the MMQBA halts in an accepting state exactly after reading symbol \( \sigma_i \), so \( P(E_i) = \alpha_i \). Since computation proceeds from the non-halting state at each step, the state at time \( i \) depends on prior steps, and thus the events \( E_1, E_2, \dots \) are not independent.

Define disjoint events \( F_n := E_n \cap \bigcap_{i=1}^{n-1} E_i^c \), denoting acceptance occurs first at step \( n \). These are mutually disjoint and satisfy:
$$
\bigcup_{n=1}^{\infty} F_n = \{ \text{eventual acceptance} \}, 
P\left( \bigcup_{n=1}^{\infty} F_n \right) = \sum_{n=1}^{\infty} P(F_n)$$ 
We have 
$$
P(F_n) = \left( \prod_{i=1}^{n-1}(1 - \alpha_i) \right) \cdot \alpha_n, 
 \text{ So }
\sum_{n=1}^\infty P(F_n) = \sum_{n=1}^\infty \left( \prod_{i=1}^{n-1}(1 - \alpha_i) \cdot \alpha_n \right)
$$
Now assume the MMQBA visits accepting states infinitely often, i.e., \( \alpha_i > 0 \) for infinitely many \( i \). Then:
$$
\prod_{i=1}^\infty (1 - \alpha_i) = 0 \;\Rightarrow\;
P\left( \bigcap_{i=1}^\infty E_i^c \right) = 0 \;\Rightarrow\;
P\left( \bigcup_{n=1}^\infty E_n \right) = 1 - P\left( \bigcap_{n=1}^\infty E_n^c \right) = 1.
$$
Note that \( \bigcup_{n=1}^\infty E_n = \bigcup_{n=1}^\infty F_n \), since the \( F_n \) form a disjoint partition. Therefore,
$$
P\left( \bigcup_{n=1}^\infty F_n \right) = \sum_{n=1}^\infty P(F_n) = 1 \quad \Rightarrow \quad \sum_{i=1}^\infty \alpha_i = 1.
$$
Also, since \( \sum_{j=1}^\infty \rho_j = 0 \), it follows that \( \rho_j = 0 \) for all \( j \), i.e., no amplitude is lost to rejection. Hence,
$
\text{rej}_j = 0 \text{ for all } j \; \text{ and } \|\psi'_j\|^2 = 1 - \text{acc}_j.
$. Since \( \alpha_j > 0 \) for infinitely many \( j \), and the disjoint union \( \bigcup_{n=1}^\infty F_n \) implies
$
\sum_{j=1}^\infty \alpha_j = 1.
$

Thus, in both formulations, the cumulative acceptance probability satisfies
$$
acc_j = \sum_{i=1}^j \alpha_i \to 1\text{ as } j \to \infty,\text{ i.e., } \lim_{j \to \infty} acc_j = 1, \text{ as required}.
$$
\end{proof}


\hypertarget{no}
\noindent\textbf{Lemma~\ref{no}.}
\textit{Let $M$ be an MMQBA with basis state set $Q$ and
$S \subseteq S_{\text{non}}$ be a $\sigma$-cycle subspace.
If any basis vector $|r\rangle \in Q$ with 
$|r\rangle \notin S$, then
$P_S\!\big(V_\sigma |r\rangle\big)=0.$
}
\begin{proof}
Let $V_\sigma=(v_{s,t})_{s,t=0}^n$ be the matrix of $V_\sigma$ in the
computational basis, where
$
v_{s,t}=\langle q_s|V_\sigma|q_t\rangle .
$. Since $V_\sigma$ is unitary, each column and each row has squared norm $1$, formally,
$\sum_{s=0}^n |v_{s,t}|^2 = 1$ for all $t$, 
$\sum_{t=0}^n |v_{s,t}|^2 = 1$ for all $s$.

Let $I \subseteq \{0,\dots,n\}$ be the set of indices such that
$|q_i\rangle \in S$.
The assumption $V_\sigma |q_i\rangle \in \operatorname{span}(S)
\quad \text{for all } i \in I$
which means that $v_{s,i}=0 \text{for all } s \notin I,\; i \in I.$ Fix any $i \in I$. All nonzero entries of column $i$ lie inside rows indexed by $I$, and therefore $\sum_{s\in I} |v_{s,i}|^2 = 1.$

Now suppose, for contradiction, there exists
$r \notin I$ such that $P_S\bigl(V_\sigma |q_r\rangle\bigr)\neq 0$. Then there exists some $s_0 \in I$ with $|v_{s_0,r}|^2 > 0.$ Consider the squared norm of row $s_0$. Since $s_0 \in I$, the contribution from columns $i \in I$ already satisfies
$\sum_{i\in I} |v_{s_0,i}|^2 = 1$.
Thus, $\sum_{t=0}^n |v_{s_0,t}|^2
\ge
\sum_{i\in I} |v_{s_0,i}|^2 + |v_{s_0,r}|^2>1$ 
which contradicts the row-norm condition.
This contradiction shows that no such $r$ can exist, and therefore
$
P_S\bigl(V_\sigma |q_r\rangle\bigr)=0$ for all $r \notin I$. 
\end{proof}

\hypertarget{limit-characterization}
\noindent\textbf{Theorem~\ref{limit-characterization}.}
\textit{For every $L' \subseteq \Sigma^\omega$, there exist an MMQBA  $M'$ and a cutpoint  $p' > \tfrac12$ such that $L' = L(M',p')
$ iff there exist an MMQFA $M$ and a cutpoint $p >\tfrac12$  such that $L' = \lim\!\bigl(L(M,p)\bigr)$
}
\begin{proof}
We prove both directions of the equivalence:
$$L(M,p) = L \;\;\Longleftrightarrow\;\; L(M',p') = \lim(L).
$$
Let $M=(Q,\Sigma,\delta,q_0,Q_{acc},Q_{rej})$ be an
MMQFA with cutpoint $p>\tfrac12$, and let $L=L(M,p)$.
We construct an MMQBA
$M'=(Q',\Sigma,\delta',q'_0,Q'_{acc},Q'_{rej})$
with cutpoint $p'>\tfrac12$ such that $L(M',p')=\lim(L)$.

The construction preserves the state space and the accepting/rejecting
subspaces, i.e.,
$Q'=Q$, $Q'_{acc}=Q_{acc}$, $Q'_{rej}=Q_{rej}$, and $q'_0=q_0$.
The only difference lies in the transition structure: unlike $M'$,
the automaton $M$ includes a right end-marker $\$$, and hence the
transition function $\delta$ extends $\delta'$ by incorporating the
action of the end-marker.

We construct $M'$ so that the effect of $V_{\$}$ is distributed
over infinitely many steps. For each $\sigma \in \Sigma$, define
unitary operators $V'_\sigma$ so that the following holds.

Let $u \in \Sigma^*$ be a prefix, and let $\psi_u$ denote the
state reached by $M'$ after reading $u$. Then $\psi_u$ may have
components in both halting and non-halting subspaces. For each $u$,
a fraction $\varepsilon \in (0,1)$ of the current non-halting probability is
transferred to the halting subspaces, and this transferred part is
distributed between $S_{acc}$ and $S_{rej}$ according to the
acceptance–rejection direction induced by $V_{\$}$. The remaining $(1-\epsilon)$
probability continues in $S_{{non}}$.

This distribution defines a transformation on $S_{{non}}$
that preserves the structure of superpositions (linearity) and
the total probability (norm preservation), since an $\varepsilon$
fraction is transferred to the halting subspaces and the remaining
$(1-\varepsilon)$ fraction stays in $S_{{non}}$. Hence it
can be realized by a unitary operator on the whole space.

\medskip

\medskip

Let $x \in \lim(L)$. Then $x$ has infinitely many prefixes
\[
u_1 \prec u_2 \prec u_3 \prec \cdots
\quad\text{with}\quad
u_j \in L.
\]

Let $r_j$ denote the total non-halting probability remaining
after processing the prefix $u_j$, with $r_0 = 1$. By construction,
each time a prefix $u_j \in L$ is processed, a fraction $\varepsilon$
of the current non-halting probability is transferred to the
halting subspaces. Hence
\[
r_1 \le (1-\varepsilon), \quad
r_2 \le (1-\varepsilon)^2, \quad \dots,
\]
and in general,
$
r_j \le (1-\varepsilon)^j.
$ Since $0 < 1-\varepsilon < 1$, we have
$
\lim_{j \to \infty} r_j = 0.
$
Thus the total halting probability tends to $1$, i.e.,
$
\lim_{j \to \infty} (acc_j + rej_j) = 1.
$

Moreover, for each $j$, since $u_j \in L$, the probability
transferred to the halting subspaces at step $j$ is biased
towards acceptance. Hence the increment in acceptance at each
step is strictly greater than the corresponding increment in
rejection, and therefore for all $j$, $acc_j > rej_j.$ Taking limits and using $acc_j + rej_j \to 1$, we obtain
$
\lim_{j \to \infty} acc_j > \tfrac12,
\;
\lim_{j \to \infty} rej_j < \tfrac12.
$ Thus there exists $p' > \tfrac12$ such that $x \in L(M',p').$

\textbf{($\Leftarrow$)}
Let $M'=(Q,\Sigma,\delta',q_0,Q_{acc},Q_{rej})$ be an MMQBA with cutpoint $p'>\tfrac12$ such that $L(M',p')=\lim(L)$. We construct an MMQFA
$M=(Q,\Sigma \cup \{\$\},\delta,q_0,Q_{acc},Q_{rej})$
such that $L(M,p)=L$ for some $p>\tfrac12$.

The construction preserves the state space and the accepting/rejecting
subspaces, i.e.,
$Q'=Q$, $Q'_{acc}=Q_{acc}$, $Q'_{rej}=Q_{rej}$, and $q'_0=q_0$.
The only difference lies in the transition structure: unlike $M'$,
the automaton $M$ includes a right end-marker $\$$, and hence the
transition function $\delta$ extends $\delta'$ by incorporating the
action of the end-marker.

Fix $\varepsilon \in (0,1)$ such that $1-\varepsilon > \tfrac12$.
For each $\sigma \in \Sigma$, we need to define unitary operators $V_\sigma$
so that, during the processing of any input $u$, at most an
$\varepsilon$ fraction of the probability is transferred to the
halting subspaces, and at least $(1-\varepsilon)$ remains in
$S_{{non}}$. This specification preserves the structure of superpositions (linearity) and the total probability (norm preservation), since at most an $\varepsilon$ fraction is transferred to the halting
subspaces and at least $(1-\varepsilon)$ remains in
$S_{{non}}$. Hence such transformations can be realized
by unitary operators.

\medskip

Let $u \in \Sigma^*$ and let $\psi_u$ be the state of $M$ after
reading $u$ (before applying $V_{\$}$). We have 
$\psi_u = P_{{non}}(\psi_u) + P_{acc}(\psi_u) + P_{rej}(\psi_u),$
where $P_{{non}}, P_{acc}, P_{rej}$ are the projections onto
$S_{{non}}, S_{acc}, S_{rej}$, respectively. By construction,
$$
\|P_{{non}}(\psi_u)\|^2 \ge (1-\varepsilon) > \tfrac12,
\text{ and }
\|P_{acc}(\psi_u)\|^2 + \|P_{rej}(\psi_u)\|^2 \le \varepsilon.
$$

\medskip

Upon reading the end-marker $\$$, the transformation $V_{\$}$
transfers the entire remaining non-halting probability to the
halting subspaces. Thus,
$$
\|P_{{non}}(V_{\$}\psi_u)\|^2 = 0,
\;
\|P_{acc}(V_{\$}\psi_u)\|^2 + \|P_{rej}(V_{\$}\psi_u)\|^2 = 1.
$$

The transfer induced by $V_{\$}$ is defined so that the resulting
acceptance–rejection distribution reflects whether $u \in L$ or
$u \notin L$.

Since $\|P_{{non}}(\psi_u)\|^2 > \tfrac12$, more than half of the
total probability resides in the non-halting component before applying
$V_{\$}$. As $V_{\$}$ transfers the entire non-halting probability to the
halting subspaces according to the acceptance–rejection behavior of $u$,
this majority determines the final outcome. Hence, if $u \in L$, the
acceptance probability exceeds $\tfrac12$, and if $u \notin L$, it is at most $\tfrac12$. Therefore, there exists $p > \tfrac12$ such that
$
u \in L \;\Longleftrightarrow\; u \in L(M,p).
$

\end{proof}

\hypertarget{non}
\noindent\textbf{Theorem~\ref{non}.}
\textit{
If $L \subseteq \Sigma^*$ is a $\sigma$-block forcing language then there exists no MMQBA $M'$ with cutpoint $p' > \frac{1}{2}$ such that $L(M',p')=\lim(L)$.
}

\begin{proof}
Assume for contradiction that an MMQBA
$M'=(Q,\Sigma,\delta,q_0,Q_{acc},Q_{rej})$
recognizes $\lim(L)$. Since $L$ is a $\sigma$-block forcing language, every word
$x\in\lim(L)$ has the form
\[
x=\sigma^{k_1}u\sigma\sigma^{k_2}u\sigma\sigma^{k_3}u\sigma\cdots
\]
with $k_i\ge 1$ for all $i$, $u\in\Sigma^+$, and $u\notin\sigma^+$. Hence the symbol $\sigma$ occurs in arbitrarily long blocks, and the transition
$V_\sigma$ is applied infinitely many times during the run of $M'$ on $x$.
Recall $S_{\text{non}}=\operatorname{span}\{|q\rangle : q\notin Q_{halt}\}$ and $Q\setminus Q_{halt}$
is a finite set. For each non-halting state $\psi_{i,m}$ arising during the run, we write, 
$$
\psi_{i,m}
=
\sum_{q \in Q_{\text{non}}}
\alpha_q^{(i,m)} \, |q\rangle,
\text{ where } Q_{\text{non}} = Q \setminus Q_{\text{halt}}
$$
We define 
$$
I_{i,m}
:=
\{\, q \in Q_{\text{non}} : \alpha_q^{(i,m)} \neq 0 \,\}
$$
and set
$$
T_{i,m}
:=
\operatorname{span}
\{\, |q\rangle : q \in I_{i,m} \,\}
\subseteq S_{\text{non}}.
$$

Each $T_{i,m}$ is a subspace spanned by a subset of the finite set
$\{ |q\rangle : q \in Q_{\text{non}} \}$.

Since $Q_{\text{non}}$ is finite, there are only finitely many possible
subsets $I_{i,m}$, and hence only finitely many possible subspaces
$T_{i,m}$. Denote these distinct subspaces by
$S_1,\dots,S_K \subseteq S_{\text{non}}$.
During the infinitely many applications of $V_\sigma$ along the
$\sigma$-blocks of $x$, we obtain an infinite sequence
$T_{1,1}, T_{1,2}, \dots, T_{2,1}, \dots$
taking values in the finite set $\{S_1,\dots,S_K\}$.

We define
$$
B_\infty
:=
\{\, q \in Q_{\text{non}} :
q \in I_{i,m} \text{ for infinitely many pairs } (i,m) \,\}.
$$
That is, $B_\infty$ consists of those basis states that occur with
nonzero amplitude in infinitely many non-halting states
$\psi_{i,m}$. Since the sequence of states is infinite and $Q_{\text{non}}$ is finite,
the set $B_\infty$ is nonempty. Now define
$$
S^{(1)}
:=
\operatorname{span}
\{\, |q\rangle : q \in B_\infty \,\}
\subseteq S_{\text{non}}.
$$
We claim that
$
V_\sigma(S^{(1)}) \subseteq S^{(1)}.
$
 
Indeed, $q \in B_\infty$.
By definition, $q$ appears in infinitely many supports $I_{i,m}$.
For each such occurrence we have
$\psi_{i,m+1}=P_{\text{non}}(V_\sigma \psi_{i,m}).$ If $V_\sigma |q\rangle$ had a nonzero component on some basis state
$r \notin B_\infty$, then $r$ would appear in infinitely many
subsequent non-halting states $\psi_{i,m+1}$,
contradicting the definition of $B_\infty$.
Hence $V_\sigma |q\rangle \in S^{(1)}$ for every $q \in B_\infty$,
and therefore
$
V_\sigma(S^{(1)}) \subseteq S^{(1)}
$. Thus $S^{(1)}$ is a nonzero $V_\sigma$-cycle subspace of
$S_{\text{non}}$.

We now claim that for every $\psi \in S^{(1)}$,
$
\|P_{acc}(V_\sigma \psi)\|^2 = 0.
$

Suppose for contradiction that there exists $\psi \in S^{(1)}$ such that
$
\|P_{acc}(V_\sigma \psi)\|^2 > 0.
$
Since $V_\sigma(S^{(1)}) \subseteq S^{(1)}$, repeated applications of
$V_\sigma$ keep the computation inside $S^{(1)}$. Thus along the infinite
input $\sigma^\omega$, acceptance leakage occurs infinitely many times,
causing the cumulative acceptance probability to exceed the cutpoint.
Hence the automaton would accept $\sigma^\omega$, contradicting
$\sigma^\omega \notin \lim(L)$. Therefore
\[
\|P_{acc}(V_\sigma \psi)\|^2 = 0
\quad \text{for all } \psi \in S^{(1)}.
\]

Note that rejecting components inside $S^{(1)}$ do not lead to a contradiction.
It may happen that $\|P_{rej}(V_\sigma \psi)\|^2 > 0$ for some
$\psi \in S^{(1)}$, which is consistent with rejection of $\sigma^\omega$.
Hence rejecting components may lie in the $V_\sigma$-cycle subspace $S^{(1)}$.

Now consider the configuration after reading the block $\sigma^{k_1}u$.
Let $\psi_0$ denote the non-halting state reached after reading
$\sigma^{k_1}$ and set $\phi = V_u(\psi_0)$.
We can write the orthogonal decomposition
$\phi = P_{acc}(\phi) + P_{rej}(\phi) + P_{\text{non}}(\phi)$.
If $P_{\text{non}}(\phi)=0$, then the computation halts after reading
$\sigma^{k_1}u$ and cannot process the remaining infinite input,
contradicting the acceptance semantics for infinite words.
Otherwise $P_{\text{non}}(\phi)\neq 0$, and the computation continues
from the non-halting residue $\psi := P_{\text{non}}(\phi)\in S_{\text{non}}$.

Consider now the terminal $\sigma$ of the first block, i.e., the prefix
$\sigma^{k_1}u\sigma$. Let $\psi_1$ denote the non-halting state reached
immediately before reading this $\sigma$. Since $\sigma^{k_1}u\sigma$ is a
prefix of some word in $L$, the acceptance–rejection bias satisfies
$\|P_{acc}(V_\sigma \psi_1)\|^2 >
\|P_{rej}(V_\sigma \psi_1)\|^2$, and in particular
$\|P_{acc}(V_\sigma \psi_1)\|^2 > 0$. After measurement the
computation continues from the non-halting residue
$\psi_1' := P_{\text{non}}(V_\sigma \psi_1) \neq 0$.

We claim that $\psi_1' \notin S^{(1)}$. Indeed, every $\psi \in S^{(1)}$
satisfies $\|P_{acc}(V_\sigma \psi)\|^2 = 0$, whereas
$V_\sigma\psi_1$ has a nonzero accepting component, hence its non-halting
residue cannot lie in $S^{(1)}$. By the no-entry lemma~\ref{no}, since
$V_\sigma(S^{(1)}) \subseteq S^{(1)}$, we have
$P_{S^{(1)}}(V_\sigma \phi) = 0$ for all $\phi \notin S^{(1)}$.
Therefore, under further applications of $V_\sigma$, the computation
starting from $\psi_1'$ can never re-enter $S^{(1)}$.

Now consider the next $\sigma$-block $\sigma^{k_2}$. As before, repeated
applications of $V_\sigma$ force the non-halting evolution to stabilize
inside a nonzero $V_\sigma$-invariant subspace $S^{(2)} \subseteq
S_{\text{non}}$ with $S^{(2)} \cap S^{(1)} = \{0\}$ and
$V_\sigma(S^{(2)}) \subseteq S^{(2)}$. Repeating this argument for each
subsequent block $\sigma^{k_m}$ yields distinct nonzero subspaces
$S^{(m)} \subseteq S_{\text{non}}$ such that
\[
V_\sigma(S^{(m)}) \subseteq S^{(m)}, \qquad
S^{(m)} \cap S^{(\ell)} = \{0\} \quad (m \neq \ell).
\]

Since the input contains infinitely many $\sigma$-blocks, this produces
infinitely many pairwise disjoint non-halting subspaces inside the
finite-dimensional space $S_{\text{non}}$, which is impossible.
\end{proof}

\hypertarget{cor1}
\noindent\textbf{Corollary~\ref{cor1}.}
\textit{
If $L \subseteq \Sigma^*$ is a $\sigma$-block forcing language, then
there exists no MMQFA $M$ with cutpoint $p > \tfrac12$ such that
$L(M,p)=L$.
}
\begin{proof}
Suppose, for contradiction, that there exists an MMQFA $M$ with
cutpoint $p>\tfrac12$ such that $L(M,p)=L$.
By Theorem~\ref{limit-characterization}, there exists an MMQBA $M'$
with cutpoint $p'>\tfrac12$ such that $L(M',p')=\lim(L)$.
However, by Theorem~\ref{non}, no MMQBA with cutpoint greater than
$\tfrac12$ can recognize $\lim(L)$ when $L$ is a $\sigma$-block forcing
language, yielding a contradiction.
\end{proof}

\hypertarget{cor2}
\noindent\textbf{Corollary~\ref{cor2}.}
\textit{
If $L\subseteq\Sigma^*$ is a finite language, then
there exist an MMQFA $M$ and an MMQBA $M'$ such that $L(M,1)=L$ and $L(M',1)=\lim(L)$.}
\begin{proof}
Let $N=\max\{|x|:x\in L\}$. Define an MMQFA
$M=(Q,\Sigma,\delta,q_0,Q_{acc},Q_{rej})$ with
$Q=\{q_0,q_1,\dots,q_N,q_{acc},q_{rej}\}$,
$Q_{acc}=\{q_{acc}\}$, $Q_{rej}=\{q_{rej}\}$,
and $q_0,\dots,q_N$ non-halting. For each position $i \in \{1,\dots,N\}$ and symbol $\sigma \in \Sigma$,
define the transition so that $V_\sigma \ket{q_{i-1}} = \ket{q_i}$.
Thus the automaton counts the length of the input up to $N$. After reading a word $x=\sigma_1\cdots\sigma_n$ with $n \le N$, the state is $\ket{q_n}$.
Upon reading the right end-marker $\$$, define
$$
V_{\$}\ket{q_n} =
\begin{cases}
\ket{q_{{acc}}}, & \text{if } x \in L,\\
\ket{q_{{rej}}}, & \text{otherwise}.
\end{cases}
$$
Thus $M$ accepts exactly the words in $L$ with probability $1$,
so $L(M,1)=L$.

Since $L$ is finite, $\lim(L)=\emptyset$. Define an MMQBA $M'=(Q',\Sigma,\delta',q_0,Q'_{acc},Q'_{rej})$
by $Q'=\{q_0\}$, $Q'_{acc}=\emptyset$, $Q'_{rej}=\emptyset$,
and $V'_\sigma\ket{q_0}=\ket{q_0}$ for all $\sigma\in\Sigma$.
Since no accepting states exist, the cumulative acceptance probability
is always $0$, and hence $L(M',1)=\emptyset=\lim(L)$.
\end{proof}
\begin{remark}
The finiteness of $L$ is not essential for the triviality of $\lim(L)$. There exist infinite languages $L\subseteq\Sigma^*$ such that
$
\lim(L)=\emptyset$. In such cases, an MMQFA recognizing $L$ with cutpoint $p>\tfrac12$ may exist,
but the corresponding MMQBA construction for $\lim(L)$ is trivial. Indeed, whenever $\lim(L)=\emptyset$, one can construct an MMQBA. 
\end{remark}
\hypertarget{union}
\noindent\textbf{Theorem~\ref{union}.}
\textit{
The class of languages recognized by MMQBA is not closed under 
union.
}
\begin{proof}
Let
\[
L_1=(ab)^\omega,
\qquad
L_2=(aab)^\omega .
\]
Both $L_1$ and $L_2$ are recognizable by MMQBA(See Example~4,5 in the Appendix). Assume for contradiction that
there exists an MMQBA $M$ with cutpoint $p>\tfrac12$ such that
\[
L(M,p)=L_1\cup L_2 .
\]

By Lemma~\ref{subspace decomposition}, the non-halting subspace decomposes as
\[
S_{\text{non}} = S^1 \oplus S^2 .
\]
The initial state must have a nonzero component in $S^2$, otherwise the
computation remains entirely in $S^1$ and no halting probability is ever
generated, implying $L(M,p)=\emptyset$. Hence we assume the run starts from a
state $\psi_0\in S^2$.

Consider the runs of $M$ on the words
\[
x_1=(ab)^\omega
\qquad\text{and}\qquad
x_2=(aab)^\omega .
\]
Since both words belong to $L(M,p)$, their runs must satisfy the Büchi
condition and the cutpoint conditions. In particular, accepting projections
must occur infinitely often, and therefore the run must visit states in $S^2$
infinitely often.

Let $\psi_{aa}$ denote the non-halting state reached after reading the prefix
$aa$ from the initial state:
\[
\psi_{aa} = P_{\text{non}}V_aP_{\text{non}}V_a(\psi_0).
\]

First consider the behavior required to accept $x_1=(ab)^\omega$.  
The prefix $aa$ does not occur in $x_1$. Hence if $\psi_{aa}$ had a nonzero
component in $S^2$, then continuing the computation from $\psi_{aa}$ along the
input $(ab)^\omega$ would again produce infinitely many accepting projections,
because states in $S^2$ generate halting probability under the transition
operators. Consequently words beginning with the prefix $aa$ followed by
$(ab)^\omega$ would also satisfy the acceptance conditions, contradicting
$L(M,p)=L_1\cup L_2$. Therefore it must hold that
\[
\psi_{aa}\in S^1 .
\]

Now consider the run of $M$ on $x_2=(aab)^\omega$.  
The prefix $aa$ occurs infinitely often in $x_2$. Since the word must be
accepted, the run must continue producing accepting projections after each
occurrence of this prefix. Thus the state reached after reading $aa$ must lie
in $S^2$, because only states in $S^2$ can generate halting amplitude
infinitely often. Hence
\[
\psi_{aa}\in S^2 .
\]

This yields the contradiction
\[
\psi_{aa}\in S^1 \cap S^2 ,
\]
which is impossible since $S^1\cap S^2=\{0\}$. Therefore no MMQBA can recognize
$L_1\cup L_2$, and the class of MMQBA-recognizable $\omega$-languages is not
closed under union.
\end{proof}

\hypertarget{intersection}
\noindent\textbf{Theorem~\ref{intersection}.}
\textit{
The class of languages recognized by MMQBA is not closed under intersection.
}
\begin{proof}
Let \( \Sigma = \{a,b\} \), and define the following $\omega$-languages:
\[
L_a = \lim\{\, w \in \Sigma^* \mid w \text{ has an odd number of } a\text{'s} \,\},
\]
\[
L_b = \lim\{\, w \in \Sigma^* \mid w \text{ has an odd number of } b\text{'s} \,\},
\]
and let \(L = L_a \cap L_b\).

Observe that both $L_a$ and $L_b$ are recognizable by MMQBA.
In particular, Example~5 in the Appendix presents an MMQBA
$M_a$ recognizing $L_a$. By a symmetric construction
(replacing the roles of $a$ and $b$), we obtain an MMQBA
$M_b$ recognizing $L_b$.

Assume for contradiction that the class of languages recognized by
MMQBA is closed under intersection. Then there exists an MMQBA $M$
with cutpoint $p>\tfrac12$ such that
\[
L(M,p)=L_a \cap L_b = L .
\]

Hence for all integers $k,\ell \ge 1$ we have
\[
(a^k b^\ell)^\omega \in L(M,p),
\]
while for every $t\ge1$
\[
a^{2t} b^\omega \notin L(M,p).
\]

We will show that
\[
\forall k,\ell \ge 1,\ (a^k b^\ell)^\omega \in L(M,p)
\;\Rightarrow\;
\exists t \ge 1 \text{ such that } a^{2t} b^\omega \in L(M,p),
\]
which contradicts the assumption that
$a^{2t} b^\omega \notin L(M,p)$.

\medskip

To analyze the behavior of $M$, consider the run of $M$ on the
ultimately periodic word
$
w = (a^k b^\ell)^\omega .
$ Let $\chi_i$ denote the pre-measurement state at step $i$ and
$\chi_i' = P_{\text{non}}\chi_i$ the post-measurement non-halting
state. The evolution of the run is
\[
\chi_0 \xrightarrow{a^k} \chi_1' \xrightarrow{b^\ell} \chi_2'
\xrightarrow{a^k} \chi_3' \xrightarrow{b^\ell} \chi_4' \cdots .
\]

Since $w \in L$ and $L(M,p)=L$, the run of $M$ on $w$ satisfies the
MMQBA acceptance condition. By Lemma~\ref{subspace decomposition} we have the decomposition
$
S_{\text{non}} = S^1 \oplus S^2 .
$

Consider the behavior of the automaton during the $a$-segments of the
input. During these segments the same unitary transition $V_a$ is
applied repeatedly. As in the analysis of $\sigma$-block forcing
languages (Theorem~\ref{non}), repeated applications of the same
unitary force the computation to remain inside a $V_a$-cycle subspace
$S^1 \subseteq S_{\text{non}}$ satisfying
$
V_a(S^1) \subseteq S^1 .
$

If an accepting projection occurred during an $a$-segment, then the
same accepting behavior would repeat on the input $a^\omega$.
This would imply $a^\omega \in L(M,p)$, which contradicts
$a^\omega \notin L_b$ and therefore $a^\omega \notin L$.
Hence no accepting projection occurs during the $a$-segments.

When the $b$-segment begins, the transition $V_b$ moves amplitude
from $S^1$ into the subspace $S^2$. Accepting amplitude is
generated from this subspace, so accepting projections may occur
during the $b$-segments of the run. Since the word
$(a^k b^\ell)^\omega$ is accepted, such accepting events occur
infinitely often.

Moreover, the action of $V_b$ alternates the computation between
the two subspaces. In particular,
\[
S^1 \xrightarrow{V_b} S^2 \xrightarrow{V_b} S^1 .
\]
Thus during each block $a^k b^\ell$ the run moves from $S^1$ to
$S^2$ when the $b$-segment begins, and later returns to $S^1$
under further applications of $V_b$ (for example at the end of the
$b$-block). Consequently the computation repeatedly cycles
\[
S^1 \;\longrightarrow\; S^2 \;\longrightarrow\; S^1
\;\longrightarrow\; S^2 \;\longrightarrow\; S^1 \;\cdots
\]
along the input $(a^k b^\ell)^\omega$.

This behavior does not violate the no-entry cycle argument used
for $\sigma$-block forcing languages, because the transition
$V_b$ does not preserve a single invariant subspace. Instead it
moves the state between the two subspaces $S^1$ and $S^2$, allowing
the same pair of subspaces to be reused throughout the run.

Therefore there exist infinitely many indices $j$ corresponding to
positions inside the $b$-blocks such that
\[
\|P_{acc}(\chi_j)\| > 0 .
\]
Consequently the cumulative acceptance sequence
$(acc_j)_{j\ge1}$ increases and satisfies
\[
\lim_{j\to\infty}acc_j \ge p .
\]

From the periodic runs $(a^k b^\ell)^\omega$ we know that the automaton
repeatedly reaches configurations preceding $b$-segments from which
accepting amplitude is generated. Let
\[
T=\{\, i\in\mathbb{N} \mid
\text{the configuration at step } i \text{ precedes such a } b\text{-segment} \,\}.
\]
Since such configurations occur infinitely often in the run,
the set $T$ is nonempty. Let $2t=\min(T)$. Now consider the word $a^{2t}b^\omega$.
Let $\psi_i$ denote the pre-measurement states and
$\psi_i' = P_{\text{non}}\psi_i$ the corresponding non-halting states. The run on the prefix $a^{2t}$ evolves as
\[
\psi_0 \xrightarrow{a} \psi_1' \xrightarrow{a} \psi_2'
\xrightarrow{a} \cdots \xrightarrow{a} \psi_{2t}.
\]
After reading $a^{2t}$ the computation reaches a state
$
\psi = \psi_{2t} \in S_{\text{non}} .
$ We now let the system evolve on the infinite input $b^\omega$
starting from this state. Let $\psi_0''=\psi$ and define
\[
\psi_j'' = P_{\text{non}}(V_b(\psi_{j-1}'')), \qquad j\ge1,
\]
so that
\[
\psi \xrightarrow{b} \psi_1'' \xrightarrow{b} \psi_2''
\xrightarrow{b} \cdots .
\]

During the runs of $(a^k b^\ell)^\omega$ we observed that the
accepting amplitude arises from the $S^2$ component and that
the subspace $S^2$ is stable under the transition $V_b$.
Therefore repeated applications of $V_b$ on the infinite input
$b^\omega$ continue to generate accepting amplitude. Hence accepting projections occur infinitely often along the run,
and the cumulative acceptance probability again satisfies
$
\lim_{j\to\infty}acc_j \ge p .
$ For rejection, no amplitude leaks into $Q_{rej}$ during
this evolution, so the cumulative rejection mass remains bounded.
Thus $\lim_{j\to\infty}rej_j < p.$ Furthermore the Büchi condition holds since accepting projections occur infinitely often. Therefore $a^{2t} b^\omega \in L(M,p),$ which contradicts the assumption that $a^{2t} b^\omega \notin L(M,p)$. Hence such an automaton $M$ cannot exist.
\end{proof}
\hypertarget{complement}
\noindent\textbf{Theorem~\ref{complement}.}
\textit{
The class of languages recognized by MMQBA is not closed under complement.
}
\begin{proof}
We proceed by contradiction.  
Assume that the class of languages recognized by MMQBA is closed under complement.  
Then, for any MMQBA $M_1$ recognizing the language $L_a$, there exists an MMQBA $M$ recognizing the complement language $\overline{L_a}$. The limit language $L_a$ is defined as
\[
L_a = \lim \{ w \mid \text{odd number of } a\text{'s} \}.
\]
The complement language is then given by
\[
\overline{L_a} = \Sigma^\omega \setminus L_a .
\]

Since $M$ recognizes $\overline{L_a}$ under the MMQBA acceptance condition, any word accepted by $M$ must satisfy the following conditions:

1. infinitely many indices $j$ with $\|P_{acc}(\psi_j)\|^2>0$ (Büchi condition),  
2. $\lim_{j\to\infty}acc_j \ge p'$,  
3. $\lim_{j\to\infty}rej_j < p'$.

We construct integers $k_1,k_2,k_3,\ldots$ and the corresponding prefixes
\[
x_1 = aab^{k_1}, \qquad
x_2 = aab^{k_1}aab^{k_2}, \qquad
x_3 = aab^{k_1}aab^{k_2}aab^{k_3},
\]
which extend to an infinite word
$
x = aab^{k_1}aab^{k_2}aab^{k_3}\cdots \in \Sigma^\omega .
$

Since $M$ recognizes $\overline{L_a}$, there exists an infinite word
whose run satisfies the MMQBA acceptance conditions. From this run we
can choose integers $k_1,k_2,k_3,\ldots$ such that after reading each
block $aab^{k_i}$ the computation reaches an accepting configuration. Consequently the run of $M$ on $x$ contains infinitely many accepting visits and satisfies the MMQBA acceptance conditions. Hence $M$ accepts the infinite word $x$.

Now consider the prefixes of $x$ that end at the first $a$ of each block.
These prefixes are
$$
a,\quad
aab^{k_1}a,\quad
aab^{k_1}aab^{k_2}a,\quad \ldots
$$
The numbers of occurrences of $a$ in these prefixes are
$
1,\quad 3,\quad 5,\quad \ldots
$ respectively. Hence there are infinitely many prefixes of $x$ containing
an odd number of $a$'s, which implies that $x \in L_a$. Thus $M$ accepts a word that belongs to $L_a$, contradicting the
assumption that $M$ recognizes $\overline{L_a}$. Therefore the class of
languages recognized by MMQBA is not closed under complementation.
\end{proof}
\end{document}